\documentclass[final,3p,times]{elsarticle}

\usepackage{amssymb}   % 数学符号表
\usepackage{caption}   % 图编号
\usepackage{subfigure} % 子图
\usepackage{amsmath}   % 公式
\usepackage{setspace}  % 公式间距
\usepackage{algorithm} % 伪代码
\usepackage{algpseudocode} % 伪代码
\usepackage{makecell}  % 表格
\usepackage{multirow}  % 表格
\usepackage[justification=centering]{caption} % 图标标题居中
\usepackage{multicol}
\usepackage{stfloats}  % h-放在此处;t-放在顶端;b-放在底端;p-在本页
\usepackage{color}     % 对正文的修改部分做高亮处理
\usepackage{amsthm}    % 证明结尾的白色方框
\usepackage{threeparttable}
\usepackage{url}
\usepackage{dcolumn}

\usepackage{algorithm}
\usepackage{algpseudocode}

\newtheorem{definition}{Definition}
\newtheorem{example}{Example}
\newtheorem{theorem}{Theorem}

\begin{document}
	%\listoffigures
	%\captionsetup[figure]{labelfont={bf},labelformat={default},labelsep=period,name={Fig.}} % 图
	%\captionsetup[table]{labelfont={bf},labelformat={default},labelsep=period,name={Table}} % 表

\begin{frontmatter}
\title{Repetitive nonoverlapping sequential pattern mining}

\author[1]{Meng Geng} 

\author[1]{Youxi Wu}	
	
\author[2]{Yan Li}
		
\author[3]{Jing Liu}
	
\author[4]{Philippe Fournier-Viger}

\author[5]{Xingquan Zhu}

\author[6]{Xindong Wu}
		
		
		\address[1]{School of Artificial Intelligence, Hebei University of Technology, Tianjin 300401, China}
		
		\address[2]{School of Economics and Management, Hebei University of Technology, Tianjin 300401, China}
		
		\address[3]{State Key Laboratory of Reliability and Intelligence of Electrical Equipment, Hebei University of Technology, Tianjin 300401, China}
  
		\address[4]{Shenzhen University, Shenzhen, China}
  
		\address[5]{Department of Computer \& Electrical Engineering and Computer Science, Florida Atlantic University, FL 33431, USA}
		
		\address[6]{Research Center for Knowledge Engineering at the Zhejiang Lab, Hangzhou 311121, China}
		
\begin{abstract}
Sequential pattern mining (SPM) is an important branch of knowledge discovery that aims to mine frequent sub-sequences (patterns) in a sequential database. Various SPM methods have been investigated, and most of them are classical SPM methods, since these methods only consider whether or not a given pattern occurs within a sequence. Classical SPM can only find the common features of sequences, but it ignores the number of occurrences of the pattern in each sequence, i.e., the degree of interest of specific users. To solve this problem, this paper addresses the issue of repetitive nonoverlapping sequential pattern (RNP) mining and proposes the RNP-Miner algorithm. To reduce the number of candidate patterns, RNP-Miner adopts an itemset pattern join strategy. To improve the efficiency of support calculation, RNP-Miner utilizes the candidate support calculation algorithm based on the position dictionary. To validate the performance of RNP-Miner, 10 competitive algorithms and 20 sequence databases were selected. The experimental results verify that RNP-Miner outperforms the other algorithms, and using RNPs can achieve a better clustering performance than raw data and classical frequent patterns. All the algorithms were developed using the PyCharm environment and can be downloaded from https://github.com/wuc567/Pattern-Mining/tree/master/RNP-Miner.

\end{abstract}
		
%\newpage

\clearpage

		\end{frontmatter}
		
\section{Introduction}
In the era of big data, data mining refers to the process of searching for hidden information in a large amount of data \cite{2tmis2022}. A sequence database is a kind of big data that consists of multiple sequences. Each sequence has an ordered list of itemsets, and an itemset represents a collection of simultaneously occurring items \cite{3tkde2020}. Sequential pattern mining (SPM) is a well-developed data mining research field that aims to discover useful but potential patterns from a sequence database \cite{4tkde2020}. A variety of SPM methods have been derived for different mining requirements, including tri-partition SPM \cite{5ins2020tri}, negative SPM \cite{7tnnls2019neg}, high-utility SPM \cite{9tkde2021utility, 10tkdd2022utility}, order-preserving SPM \cite{11tcyb2022OPP, 12tkde2022OPR}, episode SPM \cite{13episodepattern, 14episoderule}, spatial co-location SPM \cite{15tkdespatial2018, 16spatial2022}, and contrast SPM \cite{17tcyb2021constrast, 18jessica2020}.

A sequence can be divided into two types: complex general sequence and simple sequence. For complex general sequences, each sequence is an order list of itemsets, and an itemset represents a collection of simultaneously occurring items \cite{21tkde2019Truong}. Complex general sequences are often used to represent customer purchase behaviors. For example, assuming  customer 1 first purchased items ``\textit{a}'' and ``\textit{c}'', then bought ``\textit{a}'' and ``\textit{c}'' again, then purchased ``\textit{c}'', then bought ``\textit{a}'', ``\textit{c}'', and ``\textit{e}'', and finally purchased ``\textit{c}''. Thus, the shopping sequence of customer 1 is $ \textbf{s}_1 $ = [\textit{ac}][\textit{ac}][\textit{c}][\textit{ace}][\textit{c}]. For simple sequences, a sequence only consists of items \cite{61li2022apind, 62shi2020apin}. Simple sequences are often used in the field of bioinformatics, such as DNA sequence and protein sequence. For example, sequence fragment \textbf{s}=\textit{attaaa} in SARS-CoV-2 virus is a simple sequence. Simple sequence can be seen as a special case of a complex general sequence, since each itemset has only one item. %For example, \textbf{s}=\textit{attaaa} can be written as \textbf{s}=[\textit{a}][\textit{t}][\textit{t}][\textit{a}][\textit{a}][\textit{a}]. 

Now, one of the disadvantages of classical SPM for complex general sequences is that it only considers whether or not a given pattern occurs within a sequence, but ignores the number of occurrences of the pattern in a sequence \cite{19classicalSPM, 20tkdd2019gan}. For example, the support of pattern [\textit{ac}] in sequence [\textit{ac}][\textit{ac}][\textit{c}][\textit{ace}][\textit{c}] is one according to classical SPM, despite the fact that pattern [\textit{ac}] occurs more than once in this sequence. Thus, classical SPM ignores the repetition of the pattern in each sequence, which can lead to the loss of some potential patterns. Classical SPM can discover the common purchase behaviors of most users \cite{21IEEE2023}. For example, assuming there is a pattern \textbf{p} = [\textit{A}][\textit{B}][\textit{CD}][\textit{E}], where \textit{A}, \textit{B}, \textit{C}, \textit{D}, and \textit{E} are buying tickets, booking hotels, purchasing food and vegetable, and booking destination taxis, respectively. If we use classical SPM method, we can use this pattern to determine whether users have tourism travel behavior. However, it is difficult for us to use this pattern to determine the degree of user interest in the travel. If pattern \textbf{p} = [\textit{A}][\textit{B}][\textit{CD}][\textit{E}] occurs ten times in a sequence, it means that a user travels ten times. Users who travel ten times clearly have a higher level of interest in tourism than those who travel once. To cluster the users according to their behavior sequences, it is easy to obtain better clustering performance based on repetitive SPM than classical SPM method, since repetitive SPM calculates the number of occurrences of a pattern, while classical SPM only considers whether a pattern occurs in a sequence. Hence, repetitive SPM can discover more useful information than classical SPM.

Most repetitive SPM methods were dedicated to discovering specific sub-sequences in simple sequences \cite{24nosep}. Gap-constraint SPM is a kind of repetitive SPM method that can avoid discovering a large number of meaningless patterns \cite{25OWSP2021tmis}. However, it is difficult to set an appropriate gap constraint when the user does not know the characteristics of the sequence. To deal with this problem, self-adaptive gap constraints were proposed to constrain the sequence adaptively with the suitable gap \cite{26SNP2022apin}. Currently, there are many versions of repetitive SPM, such as one-off SPM \cite{29HAOP2021eswa} and nonoverlapping SPM \cite{26TKDEchang}, where the one-off method means that each item can be used at most once, and the nonoverlapping method means that each item cannot be reused by the same position in the pattern, but can be reused by different positions. In this paper, we select the nonoverlapping method to deal with the repetition SPM. The reasons are as follows. First of all, the one-off method is more likely to lose feasible occurrences, while the nonoverlapping method is less likely. For example, assuming a user lives in city $X$, buys a ticket to city $Y$, books a hotel in city $Y$, buys a ticket to city $Z$, books a hotel in city $Z$, and buys a ticket to return to $X$. The corresponding sequence is [\textit{A}][\textit{B}][\textit{A}][\textit{B}][\textit{A}]. If we use pattern [\textit{A}][\textit{B}][\textit{A}] to represent a trip, then according to the one-off method, there is only one trip in [\textit{A}][\textit{B}][\textit{A}][\textit{B}][\textit{A}], since the second [\textit{A}] cannot be reused for the next trip, while according to the nonoverlapping method, there are two trips. This example shows that the nonoverlapping method can overcome the shortage of the one-off method. More importantly, the previous works have shown that the support calculation of the one-off method is an NP-hard problem \cite{29HAOP2021eswa}, while that of the nonoverlapping method is a polynomial problem \cite{Shen2017SCIS}. Hence, we select the nonoverlapping method to deal with the repetitive SPM.

Inspired by repetitive SPM in simple sequences, we focus on repetitive SPM in complex general sequences and propose a repetitive non-overlapping SPM method for complex general sequences. Example \ref {e-nonoverlapping} illustrates the difference between classical SPM and nonoverlapping SPM.

\begin{example}\label{e-nonoverlapping}
Assuming we have a sequence database $D$, as shown in Table \ref{t-database}, and a pattern \textbf{p} = $p_{1}p_{2}p_{3}$ = [$ac$][$c$][$c$].

\begin{table}[ht]
	\centering
	% \vspace{-1.2em}
	%\footnotesize{
	\scriptsize
	\caption{Sequence database $D$}
	\begin{tabular}{lc}
            \hline
		{\textit{Sid}}     & {Sequence}  \\
		%      &    &    & {constraint}          & {confidence}   \\
		\hline
		$ \textbf{s}_{1} $  & [$ac$][$ac$][$c$][$ace$][$c$]   \\
		$ \textbf{s}_{2} $  & [$bf$][$ace$][$c]$[$de$]   \\
		$ \textbf{s}_{3} $  & [$acd$][$f$][$abc$][$ce$][$acef$][$ce$]   \\
		\hline
	\end{tabular}
    \label{t-database}
\end{table}
\begin{sloppypar}
    Pattern \textbf{p} occurs in $\textbf{s}_{1}$ = $ s_{11}s_{12}s_{13}s_{14}s_{15}$ =  [$ac$][$ac$][$c$][$ace$][$c$]  and $ \textbf{s}_{3} = s_{31}s_{32}s_{33}s_{34}s_{35}s_{36}$ = [$acd$] [$f$][$abc$][$ce$][$acef$][$ce$]. Therefore, the support of \textbf{p} in $D$ is two according to classical SPM. The nonoverlapping condition means that each itemset cannot be reused by the same position in the pattern, but can be reused by different positions. In $\textbf{s}_{1}$ = [$ac$][$ac$][$c$][$ace$][$c$], $<$1,2,3$>$ and $<$2,3,4$>$ satisfy the nonoverlapping condition, since in $ < $1,2,3$ > $, $ p_{2} =$ [$c$] $\subseteq s_{12} =$ [$ac$]  and $ p_{3} = [c] \subseteq s_{13} = [c]$, that is $ s_{12} $ is used in the second position and $ s_{13} $ is used in the third position, and in $ < $2,3,4$ > $, $ p_1 = [ac] \subseteq s_{12} = [ac]$ and $ p_2 = [c] \subseteq s_{13} = [c] $, that is $ s_{12} $ is used in the first position and $ s_{13} $ is used in the second position. Thus, the support of \textbf{p} in $ \textbf{s}_{1} $ is two. Similarly, the support of \textbf{p} in $ \textbf{s}_{3}$ is also two. Hence, the support of \textbf{p} in the database by our method is four, which is different from the support by classical SPM.
\end{sloppypar}

\end{example}

Many repetitive nonoverlapping SPM algorithms, such as NOSEP \cite{24nosep}, SCP-Miner \cite{17tcyb2021constrast}, and SNP-Miner \cite{26SNP2022apin}, have been proposed to handle different tasks. However, there are two unsolved problems in these algorithms. 1. These algorithms only can mine patterns in sequences composed of items rather than itemsets. It will be more difficult and complex to generate candidate patterns in itemset sequences. 2. More importantly, these algorithms are not efficient when it comes to calculating the support, since these algorithms use linear search to find the position of the next sub-pattern in the sequence, which will cause a lot of redundant searches.

To tackle the above problems, this paper proposes a repetitive nonoverlapping sequential pattern (RNP) mining method. Compared with previous repetitive nonoverlapping SPM methods, the novelty of the RNP mining method lies in two aspects. 1. The RNP mining method discovers patterns in complex general sequences composed of itemsets. To generate candidate patterns, for simple sequences, only S-extension is used, while for complex general sequences, both S-extension and I-extension should be used. More importantly, more effective candidate pattern pruning strategies should be explored. 2. In support calculation, previous nonoverlapping SPM methods dealt with item sequences, while the RNP mining method deals with itemset sequences, which is more complex than previous methods. More importantly, the RNP mining method adopts a position dictionary to record indexes, which can effectively avoid redundant searches.

The main contributions of this paper are as follows:

% \begin{enumerate}[1.]
% \item   
1. To consider the repetition of a pattern in each sequence with itemsets, we developed RNP mining and proposed an algorithm called RNP-Miner to mine all RNPs.

% \item 
2. To generate candidate patterns, pruning of the unpromising items (PUI) strategy and pruning of the unpromising patterns (PUP) strategy were proposed. Moreover, based on the Apriori property, RNP-Miner adopted an itemset pattern join strategy to reduce the number of candidate patterns.

% \item 
3. To calculate the support, we utilized a position dictionary to store the position of each item and pattern and proposed a candidate support calculation algorithm, which calculates the supports of super-patterns based on the occurrence positions of sub-patterns and new extension itemsets.

% \item 
4. To validate the performance of RNP-Miner, 10 competitive algorithms and 20 sequence databases were selected. The experimental results verify that RNP-Miner outperforms the other algorithms, and using RNPs can achieve a better clustering performance than classical frequent patterns.
% \end{enumerate}

\section{Related work}
The goal of SPM is to find sub-sequences (patterns) with a high frequency in sequential databases, which has been widely applied \cite{31pyramidpattern, 32generator2020kais}. Sequential databases are more common in real life, such as biological sequences \cite{33biological, 34biological}, web clicks \cite{35webclick}, and customer purchase data \cite{36purchase, 37purchase}. Sequential databases usually consider the order of the items, allowing an item to occur repeatedly at different positions of a sequence \cite{38tkdd2020}. To deal with different mining requirements, various forms of SPM have been proposed: Top-\textit{k} SPM \cite{7tnnls2019neg}, closed SPM \cite{39closed2020kbs, 40closed}, and maximal SPM \cite{41max2022apin} can reduce the number of mined patterns, and incremental SPM \cite{42incremental, 43incremental} and sliding-window SPM \cite{44slidingwindow, 45slidingwindow} can solve the problem of mining in dynamic sequential databases. In addition, to discover low-frequency but high-utility patterns, high-utility SPM methods \cite{46Truong2022appl} were proposed, such as RUP-Miner \cite{47RUP-Miner} and MDUS \cite{48MDUS}, and to discover missing events, negative SPM methods \cite{49tois2021neg} were devised, such as e-NSP \cite{50cao2016ensp} and e-RNSP \cite{51Dong2020ernsp}.

According to the type of candidate pattern generation strategy, SPM algorithms are mainly divided into three types: breadth-first strategy, such as GSP \cite{52GSP}; depth-first strategy based IDList structure, such as SPAM \cite{53SPAM}, Fast \cite{54Fast}, CM-SPAM \cite{55CM-SPAM}, and SUI \cite{52SUI}; and pattern-growth strategy, such as FreeSpan \cite{56FreeSpan} and PrefixSpan \cite{57FrefixSpan}. However, most SPM methods can be seen as classical SPM, since these methods only consider whether or not a given pattern occurs within a sequence and ignore the number of occurrences of the pattern in a sequence. For example, from Example \ref{e-nonoverlapping}, according to classical SPM, the support of pattern [$ac$][$c$][$c$] is two, since it occurs in $ \textbf{s}_{1} $ and $ \textbf{s}_{3} $ but does not occur in $ \textbf{s}_{2} $. Thus, classical SPM does not consider the fact that the pattern occurs in $ \textbf{s}_{1} $ and $ \textbf{s}_{3} $ more than once. Hence, classical SPM may neglect some potentially useful patterns. 

To deal with the shortcomings of classical SPM, repetitive SPM was proposed to consider the number of occurrences of the pattern in a sequence. Current repetitive SPM methods focus on mining patterns in sequences with items, while classical SPM is mainly meant to mine patterns in sequences with itemsets. For example, a DNA sequence is a sequence with items; it can be seen as a special case of a sequence with itemsets. Repetitive SPM can be divided into four types according to different constraints: general (no condition) \cite{27nocon2014apin}, disjoint \cite{60pmdb2022six}, one-off \cite{25OWSP2021tmis}, and nonoverlapping SPM methods \cite{58ins2021NWP}. Note that disjoint SPM was called nonoverlapping SPM in \cite{51Dong2020ernsp} and \cite{28disjoint}, but there are significant differences between the two forms.  For example, in sequence [$ac$][$ac$][$c$][$ace$][$c$], pattern [$ac$][$c$][$c$] has two nonoverlapping occurrences, i.e., $<$1,2,3$>$ and $<$2,3,4$>$. However, under the disjoint condition, the pattern only has one occurrence.

%The disjoint condition means that the maximum position of an occurrence should be less than the minimum position of the next occurrence, while the nonoverlapping condition means that each itemset can be reused by different positions.   

% Note that episode mining \cite{13episodepattern} is similar to repetitive SPM. The differences between the two methods are the following: episode mining aims at mining a single sequence, while repetitive SPM can process one or more sequences; and in episode mining, the sequence database contains a sequence with time stamps, while in repetitive SPM, the sequence database does not.

Note that episode mining \cite{13episodepattern} is similar to repetitive SPM. The differences between the two methods are the following: episode mining aims at mining a single sequence, while repetitive SPM can process one or more sequences; and in episode mining, the sequence is accompanied by timestamps, while in repetitive SPM, the sequence does not.

Recently, to address different requirements, various repetitive SPM methods have been studied. To improve the efficiency and set reasonable gap constraints when users lack prior knowledge, Wang et al. \cite{26SNP2022apin} proposed the SNP-Miner algorithm, which uses an incomplete Nettree structure to calculate the support and a pattern join strategy to generate candidate patterns. Contrast SPM was designed to mine contrast patterns from positive and negative sequences that could be used as features for a sequence classification task \cite{17tcyb2021constrast}. High average utility SPM was proposed to discover low-frequency but high average utility patterns \cite{30HANP2021kbs}. Inspired by three-way decisions \cite{65threeway2022}, nonoverlapping three-way SPM was developed to mine patterns composed of strong and medium interest items \cite{6tkdd2022NTP}. Negative SPM with one-off condition was proposed to discover missing events, which were used to predict the future trend in the traffic flow \cite{8tkdd2022neg}.

In summary, all of the above repetitive SPM methods process sequences with items. However, a general sequence consists of itemsets. For example, in a transaction sequence, an item represents a product, and an itemset represents all the products purchased by a user at the same time. Hence, we present the repetitive nonoverlapping SPM for a sequence database with itemsets, which is called RNP mining.

\section{Problem definition}

Let $ \textit{I} = \{i_{1}, i_{2}, \cdots, i_{t}\} $ be a set of items. An itemset \textit{v} is a subset of items, i.e., $ v \subseteq I $. A sequence $ \textbf{s} = s_{1}s_{2}\cdots s_{n} $ is a group of itemsets, where $ s_{j} $ ($ 1 \leq j \leq n $) is an ordered itemset. The size of \textbf{s} is its number of itemsets, denoted by size(\textbf{s}). The length of \textbf{s} is the sum of the number of items in each itemset, denoted by len(\textbf{s}). A sequence database $ D = {\textbf{s}_{1}, \textbf{s}_{2}, \cdots, \textbf{s}_{d}} $ is a set of sequences, where $ \textbf{s}_{c} (1 \leq c \leq d) $ is a sequence. For example, sequence database \textit{D} shown in Table \ref{t-database} has three sequences. Sequence $ \textbf{s}_{1} $ = [$ac$][$ac$][$c$][$ace$][$c$] has five itemsets and nine items, i.e., size($ \textbf{s}_{1} $) = 5 and len($ \textbf{s}_{1} $) = 9.

% In many applications, such as sequence classification and sequence clustering, if users do not have prior knowledge, they cannot set appropriate gaps. Thus, it is necessary to mine patterns based on the actual characteristics of the sequence, i.e., mining patterns with self-adaptive gaps, which is defined in Definition \ref{definition1}.
In many applications, such as sequence classification and sequence clustering, if users lack prior knowledge, they cannot set appropriate gaps. Thus, it is necessary to mine patterns based on the actual characteristics of the sequence, i.e., mining patterns with self-adaptive gaps, as defined in Definition \ref{definition1}.

\begin{definition}\label{definition1}
\rm (Pattern with Self-Adaptive Gaps \cite{17tcyb2021constrast}) A pattern with self-adaptive gaps can be expressed as $ \textbf{p} = p_{1} \ast p_{2}\ast\cdots\ast p_{m} $ or $ \textbf{p} = p_{1}p_{2} \cdots p_{m} $, where $ p_{k} (1 \leq k \leq m) $ is an itemset that is a subset of \textit{I}, and $\ast$ represents any number of any itemsets.
\end{definition}

\begin{definition}\label{definition2}
\rm (Occurrence and Nonoverlapping Occurrence) Given a sequence $ \textbf{s}=s_{1}s_{2}\cdots s_{n} $ and a pattern $ \textbf{p} = p_{1}p_{2} \cdots p_{m} $, $ l= <l_{1},l_{2}, \cdots ,l_{m}> $ is an occurrence of pattern \textbf{p} in sequence \textbf{s} iff $ p_{1} \subseteq s_{l_{1}} $, $ p_{2} \subseteq s_{l_{2}} \cdots $, and $ p_{m} \subseteq s_{l_{m}} (1 \leq l_{1} < l_{2} < \cdots < l_{m} \leq n) $. Assuming  there is another occurrence $ l' = <l_{1}',l_{2}', \cdots ,l_{m}'> $. \textit{l} and \textit{l}' are two nonoverlapping occurrences iff $\forall$ $ 1 \leq q \leq m $, $ l_{q} \neq l_{q}' $.
\end{definition}

\begin{definition}\label{definition3}
\rm (Support \cite{30HANP2021kbs}) The support of pattern \textbf{p} in sequence \textbf{s} is the maximum number of nonoverlapping occurrences represented by $sup(\textbf{p}, \textbf{s})$. The support of pattern \textbf{p} in sequence database \textit{D} is the sum of the supports in each sequence, i.e., $ sup(\textbf{p}, D) = \sum_{c=1}^{d} sup(\textbf{p},\textbf{s}_c )$.
\end{definition}

\begin{example}\label{e-support}
    We use sequence database $D$ in Table \ref{t-database}. Assuming we have a pattern \textbf{p} = $ p_1p_2 $ = [$ac$][$c$]. All occurrences of pattern \textbf{p} in $ \textbf{s}_{1} = s_{11}s_{12}s_{13}s_{14}s_{15} $ = [$ac$][$ac$][$c$][$ace$][$c$] are $ < $1,2$ > $, $ < $1,3$ > $, $ < $1,4$ > $, $ < $1,5$ > $, $ < $2,3$ > $, $ < $2,4$ > $, $ < $2,5$ > $, and $ < $4,5$ > $. Among them, $ < $1,3$ > $ and $ < $2,3$ > $ are two overlapping occurrences, since 3 is used in the same position in the two occurrences; while $ < $1,2$ > $ and $ < $2,3$ > $ are two nonoverlapping occurrences, since 2 is used in different positions in the two occurrences. Moreover, the maximum nonoverlapping occurrences are $ < $1,2$ > $, $ < $2,3$ > $, and $ < $4,5$ > $. Thus, the support of pattern \textbf{p} in sequence $ \textbf{s}_{1} $ is $ sup(\textbf{p}, \textbf{s}_{1}) $ = 3. Similarly, the supports of pattern \textbf{p} in sequences $ \textbf{s}_{2} $ and $ \textbf{s}_{3} $ are one and three, respectively. Hence, the support of pattern \textbf{p} in sequence database \textit{D} is $ sup(\textbf{p}, D) = \sum_{c=1}^{3} sup(\textbf{p},\textbf{s}_c ) $ = 3+1+3 = 7.
\end{example}

\begin{definition}\label{definition4}
\rm (RNP) In a sequence database \textit{D}, a pattern \textbf{p} is a repetitive nonoverlapping frequent sequential pattern with self-adaptive gap (RNP) iff its support in \textit{D} is no less than the user-specified threshold \textit{minsup}, i.e., $ sup(\textit{p}, D) \geq minsup$. Given \textit{D} and \textit{minsup}, our problem is to discover all RNPs.
\end{definition}

\begin{example}\label{e-rnp}
In Example \ref{e-support}, the support of pattern \textbf{p} = [$ac$][$c$] in database \textit{D} is 7. According to Definition \ref{definition4}, if \textit{minsup} = 6, then pattern \textbf{p} = [$ac$][$c$] is an RNP. In sequence database \textit{D} in Table \ref{t-database}, all RNPs are [$a$], [$c$], [$e$], [$ac$], [$a$][$c$], [$c$][$c$], [$ac$][$c$], and [$c$][$c$][$c$], which are shown in Table \ref{t-rnp}.

\begin{table}[ht]
	\centering
	% \vspace{-1.2em}
	%\footnotesize{
	\scriptsize
	\caption{All RNPs and their supports}
	\begin{tabular}{lcccc}
		 \hline
		{Pattern}	& {$ sup(\textbf{p}, \textbf{s}_{1}) $}	& {$ sup(\textbf{p}, \textbf{s}_{2}) $}	& {$ sup(\textbf{p}, \textbf{s}_{3}) $}	& {$ sup(\textbf{p}, D) $}	\\
		\hline
		{[$a$]}  & {3} & {1} & {3} & {7}   \\
		{[$c$]} & 5 & 2 & 5 & 12 \\
		{[$e$]} & 1 & 2 & 3 & 6 \\
		{[$ac$]} & 3 & 1 & 3 & 7 \\
		{[$a$][$c$]} & 3 & 1 & 3 & 7 \\
		{[$c$][$c$]} & 4 & 1 & 4 & 9 \\
		{[$ac$][$c$]} & 3 & 1 & 3 & 7 \\
		{[$c$][$c$][$c$]} & 3 & 0 & 3 & 6 \\
		\hline
	\end{tabular}
    \label{t-rnp}
\end{table}
\end{example}

\section{Proposed algorithm}
This section describes the proposed mining algorithm (RNP-Miner) for discovering nonoverlapping frequent patterns in a sequence database. Fig. \ref{framework} shows the framework of RNP-Miner.

%This section describes the proposed mining algorithm (RNP-Miner) for discovering nonoverlapping frequent patterns in a sequence database. There are three major steps in RNP-Miner: data preprocessing, candidate pattern generation, and support calculation. Fig. \ref{framework} shows the framework of RNP-Miner.

\begin{figure}[!ht]
	\centering
	\includegraphics[width=0.7\linewidth]{"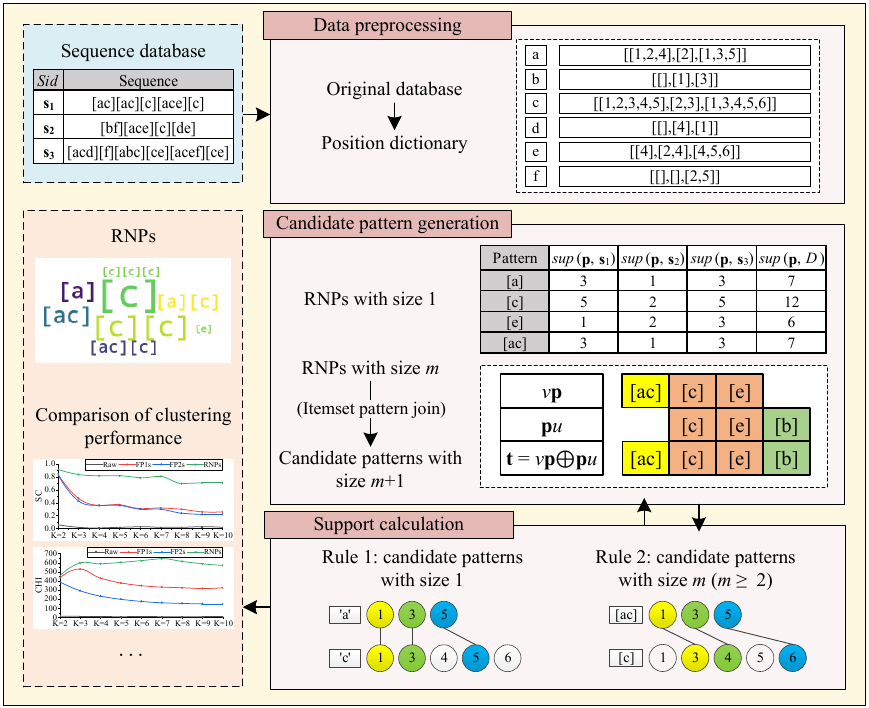"}
	\caption{Framework of RNP-Miner. RNP-Miner contains three parts: data preprocessing, candidate pattern generation, and support calculation. Data preprocessing aims to convert sequence databases into a position dictionary to reduce repeated traversal of the original database during the mining process. The candidate pattern generation stage adopts an itemset pattern join strategy to reduce the number of candidate patterns by connecting frequent patterns of size \textit{m} to generate candidate patterns of size \textit{m}+1. In the support calculation stage, efficient calculation methods are used to quickly calculate the support of patterns based on the occurrence position of sub-patterns.}
	\label{framework}
\end{figure}

\subsection{Data preprocessing}
\label{subsection:Data preprocessing}
To discover all RNPs, RNP-Miner has to scan the sequence database to calculate the supports of the candidate patterns. Obviously, it is inefficient to scan the original sequence database. To tackle this issue, we propose a position dictionary structure to improve the efficiency of the algorithm.

The position dictionary structure has two parts: keys and values. The key stores the item, and the value is a nested list with size \textit{d}, i.e., $ [[list_{1}] $, $ [list_{2}] $,$\cdots[list_{d}]] $, where \textit{d} is the number of sequences. If $ [list_{c}] $ is Null, then this means that the item does not occur in sequence $ \textbf{s}_{c} $; Otherwise, the positions of the item in sequence $ \textbf{s}_{c} $ are stored.

\begin{example}\label{e-positiondict}
	Assuming we have sequence database \textit{D} shown in Table \ref{t-database}. We take item ``d" as an example. [\textit{$ list_{\textit{1}} $}] is Null, since item ``d" does not occur in $ \textbf{s}_{1} $ = [$ac$][$ac$][$c$][$ace$][$c$]. [\textit{$ list_{\textit{2}} $}] is [4], since item ``d" occurs in the fourth itemset of $ \textbf{s}_{2} $ = [$bf$][$ace$][$c$][$de$]. Similarly, the value of item ``d" is [[], [4], [1]]. The position dictionary of sequence database \textit{D} is shown in Fig. \ref{positiondict}.
	
	\begin{figure}[ht]
		\centering
		\includegraphics[width=0.35\linewidth]{"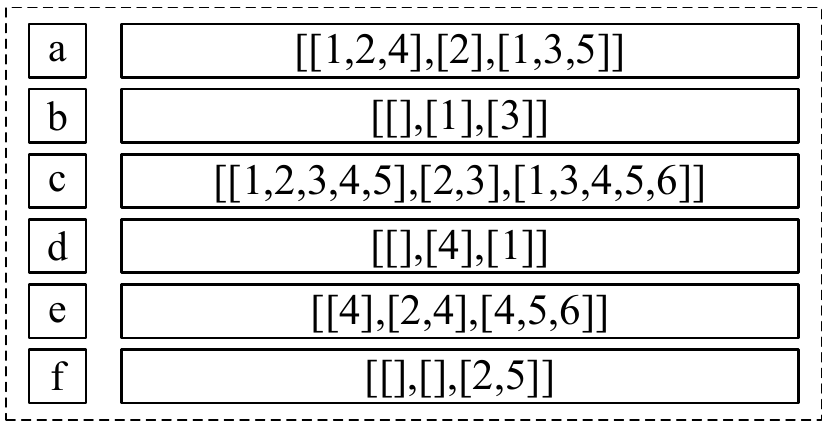"}
		\caption{Position dictionary of sequence database \textit{D}}
		\label{positiondict}
        %\vspace{-1em}
	\end{figure}

\end{example}

The pseudocode of the data processing method is given in Algorithm \ref{alg-DataPro}.

\begin{algorithm}[htb]	
	\caption{DataPro: Process the sequence database as a position dictionary.}
        \hspace*{0.02in} \leftline{{\bf Input:}
			Sequence database \textit{D}}
	\hspace*{0.02in} \leftline{{\bf Output:}
			Position dictionary \textit{S}}
	\begin{algorithmic}[1]
            \State \textit{S} $\leftarrow $\{\};
		\For {each sequence $ \textbf{s}_{c} $ in \textit{D}} 
		\For {each itemset $ s_{j} $ in $ \textbf{s}_{c} $}
		\For {each \textit{item} in $ s_{j} $}
            \If{\textit{item} not in \textit{S}.\textit{Key}}
            \State Initial \textit{S}[\textit{item}];
            \EndIf
		\State $S$[\textit{item}].list[$c$].add(\textit{j});
		\EndFor
		\EndFor
		\EndFor
		\State \Return \textit{S};
	\end{algorithmic}
    \label{alg-DataPro}     
\end{algorithm}

DataPro first initializes position dictionary \textit{S} as an empty dictionary (line 1). Next, it scans sequence database \textit{D} and reads each \textit{item} (lines 2 to 4). If \textit{item} is not in the Key of \textit{S}, then DataPro initializes \textit{S}[\textit{item}] as an empty two-dimensional array (lines 5 to 7). After that, the algorithm stores the sequence and itemset position of \textit{item} in \textit{S} (line 8). Finally, DataPro returns  position dictionary \textit{S} (line 12).

\subsection{Candidate pattern generation}
\label{subsection:Candidate pattern generation}
Many SPM methods employ the I-extension and S-extension to generate the candidate patterns \cite{47RUP-Miner}, which can be seen as a kind of enumeration tree strategy. The enumeration strategy can generate all feasible candidate patterns. However, this method will generate many redundant candidate patterns. To improve the efficiency of this method, we propose an itemset pattern join strategy based on the Apriori property. 

The definitions of the I-extension and S-extension are shown as follows.

\begin{definition}\label{definition5}
	\rm (I-Extension and S-Extension \cite{47RUP-Miner}) Assuming we have a pattern \textbf{p} and an item \textit{e}. The I-extension appends \textit{e} to the last itemset of \textbf{p}, and this is denoted by \textbf{p}$\odot$\textit{e}. The S-extension appends \textit{e} to a new itemset after the last itemset of \textbf{p}, and this is denoted by \textbf{p}$\otimes$\textit{e}.
\end{definition}

\begin{example}\label{e-ISConnection}
Assuming we have a pattern \textbf{p} = [$ac$][$c$] and an item [$e$]. According to the I-extension and S-extension, \textbf{p}$\odot$[$e$] = [$ac$][$ce$] and \textbf{p}$\otimes$[$e$] = [$ac$][$c$][$e$], respectively. It should be noted that items in an itemset are ordered, so that \textbf{p} is denoted as [$ac$][$c$] rather than [$ca$][$c$]. Hence, in the process of the I-extension, only items after the last item of this pattern can be extended. For example, assuming the last item of \textbf{p} is ``$c$''. Then, only items ``$d$'', ``$e$'', and ``$f$'' can be extended with \textbf{p} by the I-extension.
\end{example}

Obviously, the enumeration strategy will generate an exponential number of candidate patterns. To tackle this issue, we propose two pruning strategies based on the anti-monotonicity property. 

\begin{theorem}\label {theorem_anti-monotonicity}
The support satisfies the anti-monotonicity property. i.e. \textit{sup}(\textbf{q}, \textit{D})$\geq$\textit{sup}(\textbf{p}, \textit{D}), where pattern \textbf{q} is a sub-pattern of pattern \textbf{p}.
\end{theorem}

\begin{proof}
Assuming pattern $ \textbf{p}=\textit{p}_1 \textit{p}_2 \cdots \textit{p}_{\textit{m}} $. Since pattern \textbf{q} is a sub-pattern of pattern \textbf{p}, \textbf{q} can be written as $ \textit{p}_{\textit{a}} \textit{p}_{\textit{a}+1} \cdots \textit{p}_{\textit{b}} $, where \textit{a}$\geq$1 and \textit{m}$\geq$\textit{b}$\geq$\textit{a}. Suppose \textit{sup}(\textbf{p}, \textit{D})=\textit{k}, which means that there are \textit{k} nonoverlapping occurrences for \textbf{p} in database \textit{D}, i.e. $<\textit{l}_{1,1}, \textit{l}_{1,2},\cdots, \textit{l}_{1,\textit{m}}>$, $<\textit{l}_{2,1}, \textit{l}_{2,2},\cdots, \textit{l}_{2,\textit{m}}>$, $\cdots$, $<\textit{l}_{\textit{k},1}, \textit{l}_{\textit{k},2},\cdots, \textit{l}_{\textit{k},\textit{m}}>$. Therefore, the \textit{k} occurrences $<\textit{l}_{1,\textit{a}}, \textit{l}_{1,\textit{a}+1},\cdots, \textit{l}_{1,\textit{b}}>$, $<\textit{l}_{2,\textit{a}}, \textit{l}_{2,\textit{a}+1},\cdots, \textit{l}_{2,\textit{b}}>$, $\cdots$, $<\textit{l}_{\textit{k},\textit{a}}, \textit{l}_{\textit{k},\textit{a}+1},\cdots, \textit{l}_k,\textit{b}>$ are nonoverlapping. Thus, there are at least \textit{k} nonoverlapping occurrences for \textbf{q} in database \textit{D}. Hence, \textit{sup}(\textbf{q}, \textit{D})$\geq$\textit{sup}(\textbf{p}, \textit{D}).
\end {proof}

\textbf{Pruning Strategy 1 (Pruning of the unpromising items (PUI) strategy):} Given an item \textit{i}, if $ sup(i, D) < minsup $, then item \textit{i} can be pruned, and it cannot be extended to other candidate patterns.

% \textbf{Pruning Strategy 2 (Pruning of the unpromising patterns (PUP) strategy):} If a pattern \textbf{p} is not an RNP (i.e., $ sup(\textbf{p}, D) < minsup $), then \textbf{p} and its super-patterns can be pruned. 

\textbf{Pruning Strategy 2 (Pruning of the unpromising patterns (PUP) strategy):} Given a pattern \textbf{p}, if $ sup(\textbf{p}, D) < minsup $, then \textbf{p} and its super-patterns can be pruned. 

\begin{example}\label{e-pruning}
	In this example, we use sequence database \textit{D}, which is shown in Table \ref{t-database}, and \textit{minsup} = 6 to illustrate the principles of the PUI and PUP strategies. Fig. \ref{prunning} shows the candidate patterns generated by the enumeration strategy.
	
	\begin{figure}[ht]
		\centering
		\includegraphics[width=0.7\linewidth]{"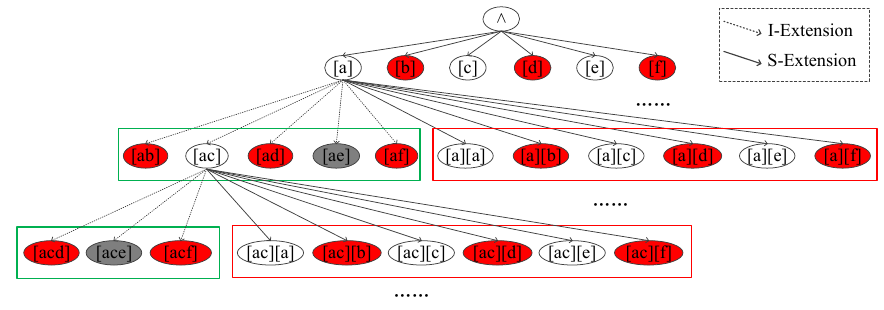"}
		\caption{The candidate pattern tree generated by the enumeration strategy. The red nodes represent the patterns that are pruned by the PUI strategy, and the grey nodes represent the patterns that are pruned by the PUP strategy.}
		\label{prunning}
	\end{figure}
	
	First, we calculate the support of each item: sup(``$a$'', $D$) = 7, sup(``$b$'', $D$) = 2, sup(``$c$'', $D$) = 12, sup(``$d$'', $D$) = 2, sup(``$e$'', $D$) = 6, and sup(``$f$'', $D$) = 3. According to the PUI strategy, ``$b$'', ``$d$'', and ``$f$'' are pruned and are not extended to create more candidate patterns, since their support values are lower than $minsup$. For example, pattern [$ab$] can be pruned, since it contains item ``$b$''. 	Then, we calculate the support of each candidate pattern generated by pattern [$a$]. According to the PUP strategy, [$ae$] is pruned, since it is not an RNP. More importantly, all its super-patterns are not RNPs.
	
\end{example}

As shown in Fig. \ref{prunning}, the PUI and PUP strategies can prune many unpromising candidate patterns. However, this candidate pattern generation method will also generate many redundant candidate patterns. To further prune the candidate patterns, we adopt an itemset pattern join strategy based on the Apriori property.

\begin{definition}\label{definition6}
	\rm (Prefix Pattern, Suffix Pattern, and Itemset Pattern Join) Assuming we have a pattern $ \textbf{p} = p_{1}p_{2}\cdots p_{m} $, and itemsets \textit{u} and \textit{v}. If $ \textbf{q} = \textbf{p}u = p_{1}p_{2}\cdots p_{m}u $, then \textbf{p} is called the prefix pattern of \textbf{q} and is denoted by \textit{Prefix}(\textbf{q}) = \textbf{p}. Similarly, if $ \textbf{d} = v\textbf{p} = vp_{1}p_{2}\cdots p_{m} $, then \textbf{p} is called the suffix pattern of \textbf{d} and is denoted by \textit{Suffix}(\textbf{d}) = \textbf{p}. We obtain a new pattern \textbf{t} of size \textit{m}+2 that is denoted by \textbf{t} = \textbf{d}$\oplus$\textbf{q} = \textit{v}\textbf{p}\textit{u}, and this process is called itemset pattern join.
\end{definition}

For example, given patterns $ \textbf{p}_{1} $ = [$ac$][$c$] and $ \textbf{p}_{2} $ = [$c$][$e$], we know that \textit{Prefix}$ (\textbf{p}_{2}) $ = [$c$] and \textit{Suffix}$ (\textbf{p}_{1}) $ = [$c$]. Thus, pattern \textbf{p} = $ \textbf{p}_{1} \oplus \textbf{p}_{2} $ = [$ac$][$c$][$e$] can be generated by the itemset pattern join strategy, since \textit{Prefix}$ (\textbf{p}_{2}) $ = \textit{Suffix}$ (\textbf{p}_{1})$.

\begin{theorem}\label {theorem1}
 Assuming $FP_{m}$ is the RNPs set with size \textit{m}. The candidate pattern set $ C_{m+1} $ can be generated from $FP_m$ using the itemset pattern join strategy. We can say that $ FP_{m+1}$ $ \subseteq C_{m+1} $, i.e., the itemset pattern join strategy is complete.
\end{theorem}

\begin{proof}
The proof is by contradiction. Assuming pattern $ \textbf{p}_{m_{ij}} $ is an RNP with size m+1, but $ \textbf{p}_{m_{ij}} $ is not in the set $ C_{m+1} $. The prefix and suffix patterns of pattern $ \textbf{p}_{m_{ij}} $ are $ \textbf{p}_{m_{j}} $ and $ \textbf{p}_{m_{i}} $, respectively, and these patterns are also RNPs. According to Definition \ref{definition6}, $ \textbf{p}_{m_{ij}} $ can be generated by the itemset pattern join strategy with $ \textbf{p}_{m_{i}} $ and $ \textbf{p}_{m_{j}} $. Thus, $ \textbf{p}_{m_{ij}} $ is in the set $ C_{m+1} $, which contradicts the hypothesis.
\end {proof}

\begin{example}\label{e-patternjoin}
	We use sequence database $D$ in Table  \ref{t-database} and set $minsup$ = 6. 38 candidate patterns are generated using the enumeration tree with the PUI and PUP strategies. However, only 29 candidate patterns are generated using the itemset pattern join strategy based on the PUI and PUP strategies. Hence, the itemset pattern join strategy based on the PUI and PUP strategies is more efficient.
\end{example}

\subsection{Support calculation}
\label{subsection:Support calculation}
% To avoid multiple scans of the database and redundant calculations, we propose a candidate support calculation (CSC) algorithm, which calculates the supports of super-patterns based on the occurrence positions of sub-patterns. The related definitions and the principles of this algorithm are given as follows.
To avoid multiple scans of the database and redundant calculations, we propose a candidate support calculation (CSC) algorithm, which calculates the supports of super-patterns based on the occurrence positions of sub-patterns. The related definitions and the principles of CSC are given as follows.

\begin{definition}\label{definition7}
	\rm (Occurrence Position) Assuming we have a sequence \textbf{s} = $ s_{1}s_{2}\cdots s_{j} $ and a pattern \textbf{p} = $ p_{1}p_{2}\cdots p_{m-1}p_{m} $, and there are \textit{k} nonoverlapping occurrences of pattern \textbf{p} in sequence \textbf{s}: $ \{<l_1^1 $, $ l_2^1 $, $\cdots$, $ l_{m-1}^1 $, $ l_m^1> $, $ <l_1^2 $, $ l_2^2 $, $\cdots$, $ l_{m-1}^2 $, $ l_m^2> $, $\cdots$, $ <l_1^k $, $ l_2^k $, $\cdots$, $ l_{m-1}^k $, $ l_m^k>\} $. Thus, the occurrence positions of pattern \textbf{p} in sequence \textbf{s} can be represented by \textit{Po}(\textbf{p}, \textbf{s}) = [$ l_m^1 $, $ l_m^2 $, $\cdots$, $ l_m^k $]. Assuming that the given sequence database \textit{D} contains \textit{d} sequences, the occurrence position of pattern \textbf{p} in sequence database \textit{D} can be recorded as a set of occurrence positions in each sequence, denoted as \textit{Po}(\textbf{p}, \textit{D}) = { \textit{Po}(\textbf{p}, $\textbf{s}_1$), \textit{Po}(\textbf{p}, $\textbf{s}_2$), $\cdots$, \textit{Po}(\textbf{p}, $\textbf{s}_d)$}.
\end{definition}

\begin{example}\label{e-po}
    % Assuming we have a sequence database $D$ and a pattern \textbf{p} = [$ac$][$c$] shown in Table \ref{t-database}. It is easy to see that the nonoverlapping occurrences of pattern \textbf{p} = [$ac$][$c$] in $ \textbf{s}_{1} $ are \{$ < $1,2$ > $, $ < $2,3$ > $, $ < $4,5$ > $\}. Thus, the occurrence positions of pattern \textbf{p} in sequence $\textbf{s}_1$ can be represented by \textit{Po}(\textbf{p}, $\textbf{s}_1$) =[2,3,5]. Similarly, the nonoverlapping occurrences of pattern [$ac$][$c$] in $ \textbf{s}_{2} $ and $ \textbf{s}_{3} $ are \{$ < $2,3$ > $\} and \{$ < $1,3$ > $, $ < $3,4$ > $, $ < $5,6$ > $\}, respectively. Hence, Po(\textbf{p}, D) = \{[2,3,5], [3], [3,4,6]\}.
    Assuming we have a sequence database $D$ and a pattern \textbf{p} = [$ac$][$c$] shown in Table \ref{t-database}. It is easy to see that the nonoverlapping occurrences of pattern \textbf{p} = [$ac$][$c$] in $ \textbf{s}_{1} $ are \{$ < $1,2$ > $, $ < $2,3$ > $, $ < $4,5$ > $\}. Thus, the occurrence positions of \textbf{p} in $\textbf{s}_1$ can be represented by \textit{Po}(\textbf{p}, $\textbf{s}_1$) =[2,3,5]. Similarly, the nonoverlapping occurrences of \textbf{p} in $ \textbf{s}_{2} $ and $ \textbf{s}_{3} $ are \{$ < $2,3$ > $\} and \{$ < $1,3$ > $, $ < $3,4$ > $, $ < $5,6$ > $\}, respectively. Hence, Po(\textbf{p}, D) = \{[2,3,5], [3], [3,4,6]\}.
\end{example}

To effectively calculate the support of the pattern, we propose two rules, one for patterns with a size of 1 and one for patterns with size \textit{m} (\textit{m} $\geq$ 2). 

\textbf{Rule 1:} Assuming size(\textbf{p}) = 1. In this case, the occurrence positions of \textbf{p} are the common elements (they can be seen as intersections in sets) of the positions of each item in \textbf{p}, and the support of \textbf{p} is the number of occurrence positions. 

\textbf{Rule 2:} Assuming size(\textbf{p}) = \textit{m} (\textit{m} $\geq$ 2), pattern \textbf{q} is the prefix pattern of pattern \textbf{p} (\textbf{q} = \textit{Prefix}(\textbf{p})), and \textit{e} is the last itemset of pattern \textbf{p}. Moreover, assuming  [$ i_{1}, i_{2}, \cdots, i_{k}, \cdots, i_{a} $] and [$ j_{1}, j_{2}, \cdots, j_{l}, \cdots, j_{b} $] are the occurrence positions of pattern \textbf{q} and itemset \textit{e}, respectively. We use [$ i_{1}, i_{2}, \cdots, i_{k}, \cdots, i_{a} $] and [$ j_{1}, j_{2}, \cdots, j_{l}, \cdots, j_{b} $] to calculate the support of pattern \textbf{p}. If $ j_{l} $ is the first unused element for which $ j_{l} > i_{k} $, then $ j_{l} $ is the matching position of $ i_{k} $ for pattern \textbf{p}. For each $ i_{k} $ in [$ i_{1}, i_{2}, \cdots, i_{k}, \cdots, i_{a} $], we find its matching position in [$ j_{1}, j_{2}, \cdots, j_{l}, \cdots, j_{b} $]. The results are the occurrence positions of pattern \textbf{p}, and the support of \textbf{p} is the number of occurrence positions.

\begin{example}\label{e-csc}
	We use sequence database $D$ in Table \ref{t-database}, and patterns \textbf{t} = [$ac$] and \textbf{p} = [$ac$][$c$] to illustrate the principles of Rules 1 and 2, respectively. 
	
	We calculate the occurrence positions of \textbf{t} according to Rule 1, since size(\textbf{t}) = 1. We take $ s_{3} $ as an example. The occurrence positions of items ``a" and ``c" are [1,3,5] and [1,3,4,5,6], respectively. Thus, the common elements are [1,3,5], as shown in Fig. \ref{rule1}. Hence, the occurrence positions of t in $ s_{3} $ are [1,3,5], and sup(\textbf{t}, $ s_{3} $) = 3. Similarly, the occurrence positions of \textbf{t} in database D are \{[1,2,4][2][1,3,5]\}, and sup(\textbf{t}, D) = 7.
	
		\begin{figure}[ht]
        %    \vspace{-1.5em}
		\centering
		\includegraphics[width=0.25\linewidth]{"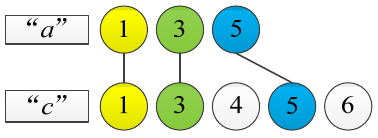"}
		\caption{The matching process of pattern \textbf{t} in $ s_{3} $. The nodes in the first level represent the positions of item ``$a$''. The nodes in the second level represent the positions of item ``$c$''. Nodes with the same color can be matched according to Rule 1.}
		\label{rule1}
	\end{figure}
	
	We calculate the occurrence positions of \textbf{p} according to Rule 2, since size(\textbf{p}) = 2. We also take $ s_{3} $ as an example. The occurrence positions of prefix pattern [$ac$] and the last itemset [$c$] are [1,3,5] and [1,3,4,5,6], respectively. Obviously, 3 in [1,3,4,5,6] is the matching position of 1 in [1,3,5], since 3 is the first unused element that is greater than 1. Hence, the occurrence positions of pattern \textbf{p} are [3,4,6], and sup(\textbf{p}, $s_{3}$) = 3, as shown in Fig. \ref{rule2}. Similarly, we can find that the occurrence positions of \textbf{p} in database D are \{[2,3,5][3][3,4,6]\}, and sup(\textbf{p}, D) = 7.
	
	\begin{figure}[ht]
        %    \vspace{-0.6em}
		\centering
		\includegraphics[width=0.25\linewidth]{"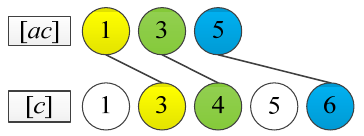"}
		\caption{The matching process of pattern \textbf{p} in $ s_{3} $. The nodes in the first level are the occurrence positions of pattern [$ac$]. The nodes in the second level are the occurrence positions of itemset [$c$]. Nodes with the same color can be matched according to Rule 2.}
		\label{rule2}
	\end{figure}
	
\end{example}

Now, Theorems \ref{Rule1-Complete} and \ref{Rule2-Complete} will prove the completeness of the CSC algorithm when size 1 and size $m$ ($m\geq$2), respectively.

\begin{theorem}\label {Rule1-Complete}
According to Rule 1, the support calculation of patterns with size 1 is complete.
\end{theorem}

\begin{proof}
This proof is by contradiction. Assume that \textit{l} is an occurrence position of the pattern, while \textit{l} is not in the common elements of the positions of each item in the pattern. Since \textit{l} is an occurrence position of the pattern, each item in this pattern occurs at position \textit{l}, i.e. \textit{l} is included in the position of all items in this pattern. Thus, \textit{l} is in the common elements of the positions of each item in the pattern, which contradicts the hypothesis.
\end {proof}

\begin{theorem}\label {Rule2-Complete}
According to Rule 2,  the support calculation of patterns with size $m$ ($m\geq$2) is complete.
\end{theorem}

\begin{proof}
Wu et al. \cite{24nosep} have shown that a complete algorithm could iteratively find the minimum occurrence. Now, we will show that Rule 2 iteratively finds the minimum occurrence. Assume that the first path is $<l_{m-1}, l_m>$, where $l_{m-1}$ is the position of sub-pattern, $l_m$ is the position of the last itemset, and $<l_{m-1}', l_m'>$ is the minimum occurrence. If $l_{m-1} \neq l_{m-1}'$, there will be two cases.

Case 1: Assume that $l_{m-1} < l_{m-1}'$. Thus, $<l_{m-1}, l_m'>$ is the minimum occurrence. This contradicts the assumption that $<l_{m-1}', l_m'>$ is the minimum occurrence.

Case 2: Assume that $l_{m-1} > l_{m-1}'$. Thus, $l_{m-1}'$ is the first position of the sub-pattern. This contradicts the fact that Rule 2 selects $l_{m-1}$ as the first position.

If $l_{m-1}=l_{m-1}'$, we assume that $<l_{m-1}, l_m>$ and $<l_{m-1}, l_m'>$ are two occurrences, and $l_m \neq l_m'$, there will be two cases.

Case 1: Assume that $l_m < l_m'$. Thus, $<l_{m-1}, l_m>$ is the minimum occurrence. This contradicts the assumption that $<l_{m-1}, l_m'>$ is the minimum occurrence.

Case 2: Assume that $l_m > l_m'$. According to Rule 2, $l_m'$ is the first unused element that can be matched. This contradicts the assumption that the first path is $<l_{m-1}, l_m>$.

In summary, we can iteratively find the minimum occurrence. Hence, according to Rule 2, the support calculation of patterns with size $m$ ($m\geq$2) is complete.
\end {proof}

\begin{theorem}\label {Correct}
The CSC algorithm is correct, i.e. each occurrence of the pattern satisfies the nonoverlapping condition.
\end{theorem}

\begin{proof}
For a pattern with size 1, we do not consider the nonoverlapping condition. For a pattern with size $m$ ($m\geq$2), its occurrence positions are obtained by matching the occurrence positions of the prefix pattern with the last itemset according to Rule 2. Assume that there are two occurrences $<l_{m-1}, l_m>$ and $<l_{m-1}', l_m'>$, and these do not satisfy the nonoverlapping condition. We only need to consider whether $l_m$ and $l_m'$ satisfy the nonoverlapping condition, since $l_{m-1}$ and $l_{m-1}'$ are the positions where the prefix pattern occurs. Only when $l_m=l_m'$, $<l_{m-1}, l_m>$ and $<l_{m-1}', l_m'>$ do not satisfy the nonoverlapping condition, but this contradicts the fact that the unused position is selected in Rule 2. Hence, the CSC algorithm is correct.
\end {proof}

The pseudocode of  CSC is given in Algorithm \ref{alg-CSC}.

\begin{algorithm}[htb]	
	\caption{CSC: Calculate the support of the pattern based on the occurrence positions.}
        \hspace*{0.02in} \leftline{{\bf Input:}
			pattern \textbf{p}, position dictionary \textit{S}, and the occurrence positions of \textbf{p} \textit{Po}}
	\hspace*{0.02in} \leftline{{\bf Output:}
			The support of \textbf{p} \textit{sup}(\textbf{p}, \textit{D})}
	\begin{algorithmic}[1]
        \State \textit{seqNum} $\leftarrow$ the number of sequence in database \textit{D};
		\If {size(\textbf{p})==1}
		\State \textbf{q} $\leftarrow$ \textbf{p}[1:len(\textbf{p})-1]; \textit{i} $\leftarrow$ \textbf{p}[len(\textbf{p})];
        \For{\textit{c}=1 to \textit{seqNum}}
        \State \textit{Po}(\textbf{p}, \textit{D}).list[\textit{c}] $\leftarrow$ \textit{Po}(\textbf{q}, \textit{D}).list[\textit{c}] $\cap$ \textit{S} [\textit{i}].list[\textit{c}];
        \EndFor
		\Else
		\State \textbf{q} $\leftarrow$ \textit{Prefix}(\textbf{p}); \textbf{t} $\leftarrow$ the last itemset of \textbf{p};
        \For{\textit{c}=1 to \textit{seqNum}}
        \State flag $\leftarrow$ 0;
        \For{\textit{i}=1 to \textit{Po}(\textbf{q},\textit{D}).list[\textit{c}].len()}
        \For{\textit{j}=\textit{flag} to \textit{Po}(\textbf{t},\textit{D}).list[\textit{c}].len()}
        \If{\textit{Po}(\textbf{t},\textit{D}).list[\textit{c}][\textit{j}] $>$ \textit{Po}(\textbf{q},\textit{D}).list[\textit{c}][\textit{i}]}
        \State \textit{Po}(\textbf{p},\textit{D}).list[\textit{c}].add(\textit{j});
        \State \textit{flag} $\leftarrow$ \textit{j}; break;
        \EndIf
        \EndFor
        \EndFor
        \EndFor
		\EndIf
		\State \textit{sup}(\textbf{p}, \textit{D}) $\leftarrow$ \textit{Po}(\textbf{p}, \textit{D}).len();
		\State \Return \textit{sup}(\textbf{p}, \textit{D});
	\end{algorithmic}
	\label{alg-CSC}
\end{algorithm}

The CSC algorithm first obtains the number of sequences in \textit{D}, denoted by \textit{SeqNum} (line 1). Then, CSC determines whether the size of input pattern \textbf{p} is 1 (line 2). If yes, according to Rule 1, the algorithm calculates the occurrence positions of pattern \textbf{p} in each sequence and stores them  into \textit{Po}(\textbf{p}, \textit{D}) (lines 3 to 6). If not, according to Rule 2, the algorithm first obtains the prefix pattern \textbf{q} of \textbf{p} and the last itemset \textbf{t} (line 8), and initializes \textit{flag} to 0 to record the usage of the occurrence position of \textbf{t} (line 10). Next, CSC traverses each occurrence position of \textbf{q} and \textbf{t} (lines 11 to 12). If the current occurrence position of \textbf{t} is greater than that of \textbf{q}, then CSC records the occurrence position of \textbf{t} in the occurrence position of \textbf{p} and set \textit{flag} to \textit{j} (lines 13 to 15). Finally, \textit{sup}(\textbf{p}, \textit{D}) is the number of items in \textit{Po}(\textbf{p}, \textit{D}) (lines 21 to 22).

\subsection{RNP-Miner}
\label{subsection:RNP-Miner}
The RNP-Miner algorithm involves the following steps.

Step 1: Use the DataPro algorithm to convert sequence database \textit{D} into position dictionary \textit{S}. 

Step 2: Calculate the support of each pattern with a size of 1, and store the RNPs in $ FP_{1} $.

Step 3: Generate the candidate pattern set $C_{m}$ from $ FP_{m-1} $ using the itemset pattern join strategy.

Step 4: For each pattern \textbf{p} in set $C_{m}$, calculate its support using the CSC algorithm. If \textit{sup}(\textbf{p}, \textit{D}) $\geq$ \textit{minsup}, store \textbf{p} in $ FP_{m}$.

Step 5: Iterate steps 3 and 4 until $ C_{m+1} $ is empty. The RNPs are stored in sets $ FP_{1} $-$ FP_{m} $.

Example \ref{e-rnpMiner} illustrates the principle of RNP-Miner.

\begin{example}\label{e-rnpMiner}
	We use  sequence database $D$ in Table \ref{t-database} and set $minsup$ = 6.

	Step 1: Convert the database D into position dictionary S shown in Fig. \ref{positiondict}.
 
	Step 2: Store the RNPs with a size of 1 in $ FP_{1} $, which is \{[$a$], [$c$], [$e$], [$ac$]\}.
	
	Step 3: Generate candidate patterns $ C_{2} $ from $ FP_{1} $ using the itemset pattern join strategy. It is noted that the prefix pattern and suffix pattern of the pattern in $ FP_{1} $ are NULL. Hence, pattern [$a$] can be extended to [$a$][$a$], [$a$][$c$], [$a$][$e$], and [$a$][$ac$]. Similarly, the number of candidate patterns generated in this step is 4$\times$4 = 16, and these candidate patterns are stored in $ C_{2} $.
	
	Step 4: Calculate the supports of the patterns in $ C_{2}$ and store the RNPs. For example, the support of pattern [$a$][$c$] is 7, which is greater than \textit {minsup}. Hence, pattern [$a$][$c$] is an RNP. Similarly, we obtain the set of RNPs with a size of 2: \{[$a$][$c$],[$c$][$c$],[$ac$][$c$]\}.
	
	Step 5: Iterate steps 3 and 4, i.e., use $ FP_{2} $ to generate $ C_{3} $ by the itemset pattern join strategy at first. For example, Suffix([$a$][$c$]) = [$c$], and Prefix([$c$][$c$]) = [$c$]. Hence, pattern [$a$][$c$] can be joined with pattern [$c$][$c$] to generate candidate patterns [$a$][$c$][$c$]. Iterating these processes, eight RNPs are discovered: [$a$], [$c$], [$e$], [$ac$], [$a$][$c$], [$c$][$c$], [$ac$][$c$], and [$c$][$c$][$c$].
\end{example}

\begin{theorem}\label {RNP-Miner-complete}
RNP-Miner is complete.
\end{theorem}

\begin{proof}
RNP-Miner has three parts: data preprocessing, candidate pattern generation, and support calculation. Data preprocessing only converts the original database into a position dictionary, and does not affect the completeness of RNP-Miner. Theorem \ref{theorem1} proves that candidate pattern generation is complete, and Theorems \ref{Rule1-Complete} and \ref{Rule2-Complete} prove that the calculation of support is complete. Hence, RNP-Miner is complete.
\end {proof}

\begin{theorem}\label {space-com}
The space complexity of RNP-Miner is $O(t\times N/ r)$, where $t$, $N$, and $r$ are the number of RNPs, the total length of database $D$, and the item set size of database $D$, respectively.
\end{theorem}

\begin{proof}
The average size of each item is $O(N/r)$, since there are $r$ items in sequence database $D$. According to Rules 1 and 2, the space complexity of the occurrence position of each pattern is also $O(N/r)$. Hence, the space complexity of RNP-Miner is $O(t$ $\times$$N/r$), since there are $t$ RNPs.
\end {proof}

\begin{theorem}\label {time-com}
The time complexity of RNP-Miner is $O(N+T \times N/r)$, where $T$ is the number of candidate patterns.
\end{theorem}

\begin{proof}
RNP-Miner has three parts: data preprocessing, candidate pattern generation, and support calculation. Firstly, to convert the sequence database into a position dictionary, we need to scan the sequence database once. Thus, the time complexity of the data preprocessing is $O(N)$. Then, we use ItemsetPatternJoin to generate candidate patterns. Assuming the size of $FP_m$ is $f_m$. Thus, the time complexity of lines 4 to 12 in Algorithm 4 is $O$(log($f_m$)), since it is a binary search. The time complexity of lines 14 to 19 in Algorithm 4 is $O(h)$, where $h$ is the average number of patterns that can be joined for each frequent pattern. Therefore, the time complexity of ItemsetPatternJoin is O($f_m \times$(log($f_m$)+h)). Since there are $T$ candidate patterns and $T\geq \sum f_m$ , the time complexity of candidate pattern generation is $O(T \times$(log($f_m$)+ $h$)). Finally, it is easy to know that the time complexity of CSC to calculate the occurrence positions of the super-pattern is $O(N/r)$, since each pattern has $N/r$ occurrence positions. For all candidate patterns, the time complexity of CSC is $O(T \times N/r)$. Hence, the time complexity of the RNP-Miner algorithm is $O(N+T \times$(log($f_m$)+ $h)+T$ $\times N/r$) = $O(N+T \times$ $N/r$) based on the fact that $f_m$ and $h$ are far less than $N$.
\end {proof}

RNP-Miner is shown in Algorithm \ref{alg-RNPMiner}.

\begin{algorithm}[htb]	
	\caption{RNP-Miner: Mine all RNPs.}
            \hspace*{0.02in} \leftline{{\bf Input:}
			sequence database \textit{D} and \textit{minsup}}
	   \hspace*{0.02in} \leftline{{\bf Output:}
			RNP set \textit{FP}}
	\begin{algorithmic}[1]
		\State \textit{S} $ \leftarrow $ DataPro(\textit{D});
		\State Calculate the support of each pattern with size 1, and store the RNPs with size 1 in $ FP_{1} $ and $ FP $ and the occurrence positions of RNPs in \textit{Po}.
		\State \textit{m} $\leftarrow$ 2;
		\While{$ FP_{m-1} $ $\neq$ NULL}
		\State $ C_{m} $ $\leftarrow$ ItemsetPatternJoin($ FP_{m-1} $);
		\For{each \textbf{p} in $ C_{m} $}
		\State \textit{sup} $\leftarrow$ CSC(\textbf{p}, \textit{Po});
		\If{\textit{sup} $\geq$ \textit{minsup}}
		\State $ FP_{m} \leftarrow FP_{m} \cup \textbf{p} $;
		\EndIf
		\EndFor
		\State $ FP \leftarrow FP \cup FP_{m} $;
		\State \textit{m} $\leftarrow$ \textit{m}+1;
		\EndWhile
		\State \Return \textit{FP};
	\end{algorithmic}
	\label{alg-RNPMiner}
\end{algorithm}

The RNP-Miner algorithm first calls the DataPro algorithm to preprocess sequence database \textit{D} to obtain the position dictionary \textit{S} (line 1). Then, the algorithm obtains the frequent patterns with size 1 and their occurrence positions (line 2). After that, RNP-Miner initializes \textit{m} to 2 (line 3), and if the frequent pattern set with size \textit{m}-1 is not empty, then the ItemsetPatternJoin algorithm (shown in Algorithm \ref{alg-join}) is used to generate candidate patterns with size \textit{m} (lines 4 to 5). For each candidate pattern \textbf{p}, RNP-Miner employs the CSC algorithm to calculate the support of pattern \textbf{p} (lines 6 to 7), and then checks condition \textit{sup}$\geq$\textit{minsup} and puts the patterns that meet to condition into a frequent pattern set with size \textit{m} (lines 8 to 10). After calculating all candidate patterns with size \textit{m}, the algorithm stores the frequent patterns with size \textit{m} into the frequent pattern set \textit{FP} and \textit{m} is incremented by one (lines 12 to 13). The algorithm repeats lines 5 to 13 until there are no frequent patterns with size \textit{m}. Finally, RNP-Miner returns the frequent pattern set \textit{FP}  (line 15).

\begin{algorithm}[!htb]
	\caption{ItemsetPatternJoin: Generate candidate patterns using itemset pattern join strategy.}
            \hspace*{0.02in} \leftline{{\bf Input:}
			RNP set \textit{FP}}
	    \hspace*{0.02in} \leftline{{\bf Output:}
			  Candidate pattern set \textit{C}}
	\begin{algorithmic}[1]
		\State \textit{C} $\leftarrow$ \{\};
		\For{each \textbf{p} in \textit{FP}}
		\State $ \textbf{p}_{suffix} \leftarrow $ \textit{suffix}(\textbf{p});
		\State \textit{min} ← 0; \textit{max} ←\textit{FP}.size(); \textit{i} ← (\textit{min}+\textit{max})/2;
        \While{(\textit{prefix}(\textit{FP}[\textit{i}]) $\neq$ $\textbf{p}_{suffix}$ \&\& \textit{max} $<$ \textit{FP}.size() \&\& \textit{min} $>$ 0 \&\& \textit{max} $>$ \textit{min}) $\|$ (\textit{prefix}(\textit{FP}[\textit{i}]) $==$ $\textbf{p}_{suffix}$ \&\& \textit{prefix}(\textit{FP}[\textit{i}-1]) $==$ $\textbf{p}_{suffix}$)}
        \If{\textit{prefix}(\textit{FP}[\textit{i}]) $\geq$ $\textbf{p}_{\textit{suffix}}$ }
        \State \textit{max} $\leftarrow$ \textit{i}-1;
        \Else
        \State  \textit{min} $\leftarrow$ \textit{i}+1;
        \EndIf
        \State \textit{i} $\leftarrow$ (\textit{min}+\textit{max})/2;
        \EndWhile
		\While{$ \textbf{p}_{suffix} $ == \textit{prefix}(\textit{FP}[\textit{i}])}
		\State \textbf{t} $\leftarrow$ \textbf{p} $\oplus$ \textbf{q};
		\State \textit{C} $\leftarrow$ \textit{C} $\cup$ \textbf{t};
		\State \textit{i} $\leftarrow$ \textit{i}+1;
		\EndWhile
		\EndFor
		\State \Return \textit{C};
	\end{algorithmic}
	\label{alg-join}
\end{algorithm}

The ItemsetPatternJoin algorithm first initializes \textit{C} as an empty set to store the generated candidate patterns (line 1). For each pattern \textbf{p} in frequent pattern set \textit{FP}, the algorithm obtains the suffix pattern of \textbf{p} and records it as $\textbf{p}_{suffix}$ (lines 2 to 3), and then uses binary search to find the first pattern in \textit{FP} whose prefix pattern is the same as $\textbf{p}_{suffix}$ (lines 4 to 12). After that, according to the itemset pattern join strategy, it generates candidate patterns and stores them in \textit{C}. The algorithm iterates this process until the prefix pattern of this pattern is not equal to $\textbf{p}_{suffix}$ (lines 13 to 17). Finally, the algorithm returns the candidate pattern set \textit{C}  (line 19).
		
\section{Experimental results and analysis}\label{section5}
To evaluate the efficiency of RNP-Miner, this paper considers the following Research Questions (RQs):

\begin{enumerate}[RQ1:]
\item Can the position dictionary improve the efficiency of support calculation?
\item Can CSC improve the computational efficiency of RNP-Miner?
\item Can pruning strategies effectively reduce the number of patterns and improve the algorithm's performance?
\item How does the itemset pattern join strategy perform compared with the classic enumeration strategy?
\item Does RNP-Miner give better performance than other classical state-of-the-art algorithms?
\item Does the value of the parameter $minsup$ affect the running time performance of RNP-Miner?
\item How good is the scalability of RNP-Miner for large-scale databases?
\item How does the clustering performance of RNPs compare with that of the raw database and classic frequent patterns?
\end{enumerate}
 
For RQ1, we use NOSEP-RNP and SNP-RNP to verify the effectiveness of the position dictionary (subsection \ref{subsection:5.2}). For RQ2, we apply DFOM-RNP and Pro-RNP to investigate the effect of CSC shown in subsection \ref{subsection:5.2}. For RQ3, we propose RNP-PUI and RNP-PUP to verify the effect of the pruning strategies (subsection \ref{subsection:5.2}). For RQ4, we propose RNP-B and RNP-D to explore the effect of the itemset pattern join strategy (subsection \ref{subsection:5.2}). For RQ5, we select PrefixSpan to explore the mining performance of RNP-Miner, as described in subsection \ref{subsection:5.2}. For RQ6, we verify the running performance with different $minsup$ values in subsection \ref{subsection:5.4}. For RQ7, we test the scalability of RNP-Miner using large-scale databases in subsection \ref{subsection:5.5}. For RQ8, we select the state-of-the-art SPM method CM-SPAM as a competitive algorithm and verify the clustering performance using RNPs (subsection \ref{subsection:5.6}).

\subsection{Databases and baseline algorithms}
\label{subsection:5.1}
To verify the performance of RNP-Miner, 20 databases were chosen, which are described in Table \ref{t-exper}. Note that $ |\sum| $ is the number of different items in the database, $ |I| $ is the sum of the sizes of all sequences, $ |s| $ is the number of sequences, Avg(\textit{I}) is the average number of itemsets per sequence, Avg(\textit{i}) is the average number of items per itemset, and $ |D| $ is the sum of the lengths of all sequences. 

\begin{table}[ht]
	\centering
	% \vspace{-1.2em}
	%\footnotesize{
	\scriptsize
	\caption{Description of databases}
    \resizebox{11cm}{!}{
	\begin{tabular}{llcccccc}
		\hline
		{Database}	& Name	& $ |\sum| $	& $ |I| $	& $ |s| $ & Avg(\textit{I}) & Avg(\textit{i}) & $ |D| $	\\
		\hline
		{SDB1} & Babysale & 6 & 3,967 & 359 & 11.05 & 2.49 & 9,870 \\
		{SDB2} & E-Shop & 283 & 5,258 & 766 & 6.84 & 9 & 47,322 \\
		{SDB3} & PM 2.5 & 580 & 1,729 & 60 & 28.82 & 20.09 & 34,743 \\
		{SDB4} & Online & 2,904 & 5,551 & 893 & 6.22 & 17.43 & 96,774 \\
		{SDB5} & OnlineRetail-Best & 7,370 & 11,367 & 1,997 & 5.59 & 9 & 102,303 \\
		{SDB6} & OnlineRetail-All & 15,722 & 5,463 & 983 & 5.56 & 21.13 & 115,406 \\
		{SDB7} & Sign & 267 & 37,958 & 730 & 52 & 1 & 37,958 \\
		{SDB8} & Leviathan & 9,024 & 197,093 & 5,834 & 33.78 & 1 & 197,093 \\
		{SDB9} & E-Shop-1k & 63 & 25 & 3 & 8.33 & 9 & 225 \\
		{SDB10} & E-Shop-5k & 145 & 142 & 23 & 6.17 & 9 & 1,278 \\
		{SDB11} & E-Shop-10k & 177 & 283 & 46 & 6.15 & 9 & 2,547 \\
		{SDB12} & E-Shop-50k & 266 & 1,415 & 210 & 6.74 & 9 & 12,735 \\
		{SDB13} & E-Shop-100k & 275 & 2,882 & 425 & 6.78 & 9 & 25,938 \\
		{SDB14} & E-Shop-500k & 291 & 14,416 & 2,085 & 6.91 & 9 & 129,744 \\
		{SDB15} & E-Shop-1000k & 303 & 28,851 & 4,193 & 6.88 & 9 & 259,659 \\
		{SDB16} & E-Shop-5000k & 312 & 144,275 & 20,979 & 6.88 & 9 & 1,298,475 \\
            {SDB17} & BC & 10 & 2,560 & 20 & 128 & 1 & 2,560\\
            {SDB18} & BF & 20 & 1,200 & 20 & 60 & 1 & 1,200\\
            {SDB19} & Credit & 18 & 180,000 & 10,000 & 18 & 1 & 180,000\\
            {SDB20} & Letter & 16 & 320,000 & 20,000 & 16 & 1 & 320,000\\
		\hline
	\end{tabular}}
	\label{t-exper}
        %\vspace{-1em}
	
 \flushleft
	\scriptsize
	\begin{enumerate}[Note 1:] 
		\item SDB1 is a sale database for infant products downloaded from https://tianchi.aliyun.com/dataset/dataDetail?dataId=45.
		\item SDB2 and SDB9-16 are click-stream databases for online shopping downloaded from http://www.philippe-fournier-viger.com/spmf/datasets/e\_shop.txt.
		\item SDB3 is a database containing the PM2.5 data of the US Embassy in Beijing downloaded from https://archive.ics.uci.edu/ml/datasets/Beijing+PM2.5+Data.
		\item SDB4 is a sale database for large furniture downloaded from https://www.kaggle.com/datasets/carrie1/ecommerce-data.
		\item SDB5-9 are real databases downloaded from http://www.philippe-fournier-viger.com/spmf/index.php?link=datasets.php.
        \item SDB17 and SDB18 are time series databases downloaded from http://www.timeseriesclassification.com/dataset.php.
        \item SDB19 is a credit card transaction database downloaded from http://www.kesci.cpm/mw/dataset/5b56a592fc7e9000103c0442.
        \item SDB20 is a real database downloaded from http://archive.ics.uci.edu/ml.
	\end{enumerate}
% \vspace{-1.2em}
\end{table}

To validate the performance of RNP-Miner, ten competitive algorithms were selected:

% \begin{enumerate} [1.]
1. NOSEP-RNP \cite{24nosep} and SNP-RNP \cite{26SNP2022apin}: To verify the efficiency of the position dictionary in support calculation, we select two state-of-the-art algorithms without position dictionary: NOSEP-RNP \cite{24nosep} and SNP-RNP \cite{26SNP2022apin} algorithms.

2. DFOM-RNP \cite{30HANP2021kbs} and Pro-RNP \cite{17tcyb2021constrast}: To evaluate the superiority of the CSC algorithm in calculating the support, we design the DFOM-RNP and Pro-RNP algorithms, which utilize the DFOM \cite{30HANP2021kbs} and MatchingPro \cite{17tcyb2021constrast} algorithms to calculate the support, respectively.

3. RNP-PUI and RNP-PUP: To validate the effectiveness of the pruning strategies, we propose RNP-PUI and RNP-PUP, which apply only the PUI and PUP strategies, respectively.

4. RNP-B and RNP-D: To verify the efficiency of the itemset pattern join strategy, we propose the RNP-B and RNP-D algorithms, which apply the enumeration strategy based on the breadth-first strategy and the depth-first strategy, respectively, in the candidate pattern generation step.

5. PrefixSpan \cite{57FrefixSpan} and CM-SPAM \cite{55CM-SPAM}: To compare the mining ability, we employ two classical state-of-the-art SPM methods: PrefixSpan \cite{57FrefixSpan} and CM-SPAM \cite{55CM-SPAM}.%, which ignore the repetitions of patterns in a sequence.
% \end{enumerate} 
	
All experiments are conducted on a computer with an Intel Xeon(R) Silver 4210 CPU @ 2.20 GHz*20 with 64 GB of memory and the Ubuntu operating system. All the algorithms were developed using the PyCharm environment and can be downloaded from https://github.com/wuc567/Pattern-Mining/tree/master/RNP-Miner.

\subsection{Mining performance and analysis}
\label{subsection:5.2}

To validate the performance of RNP-Miner, we used nine competitive algorithms and performed experiments on SDB1-SDB8. Since NOSEP-RNP, SNP-RNP, DFOM-RNP, Pro-RNP, RNP-PUI, RNP-PUP, RNP-B, and RNP-D are complete nonoverlapping SPM, the mining results are the same for these algorithms, and they can discover 1,014, 288, 736, 50, 8, 5, 461, and 13 patterns when $minsup$ = 800, 1,500, 170, 200, 3,000, 2,000, 300, and 4,000 on SDB1-SDB8, respectively. To make a fair comparison, we set different \textit{minsup} in PrefixSpan, since the definition of support in classical SPM is different from ours. We set  \textit{minsup} of PrefixSpan to 310, 345, 56, 170, 900, 240, 300, and 2,400 on SDB1-SDB8, respectively. Comparisons of running time, number of candidate patterns,  and memory usage are shown in Figs. \ref{eff-time}-\ref{eff-memory}, respectively. 

\begin{figure}[ht]
    \centering
    \includegraphics[width=0.75\linewidth]{"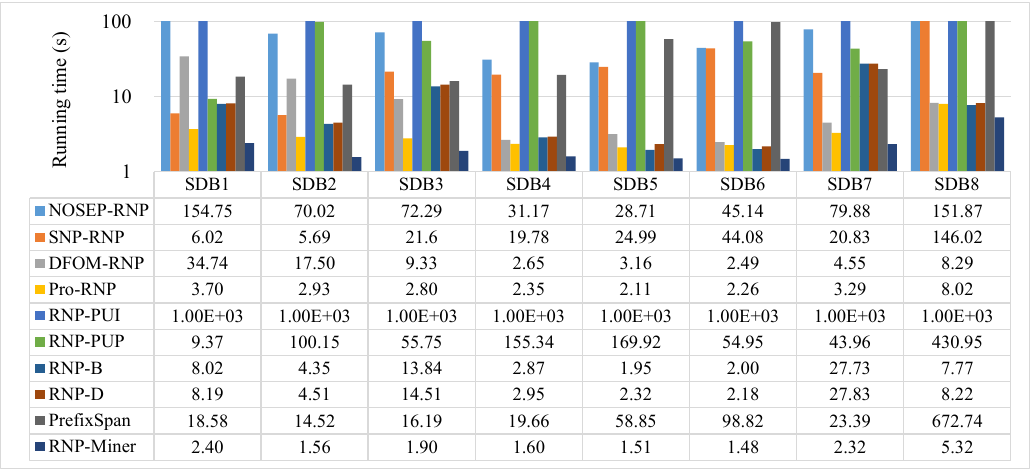"}
    \vspace{-0.9em}
    \caption{Comparison of running time}
    \label{eff-time}
    \vspace{-1.2em}
\end{figure}
 
\begin{figure}[ht]
    \centering
    \includegraphics[width=0.75\linewidth]{"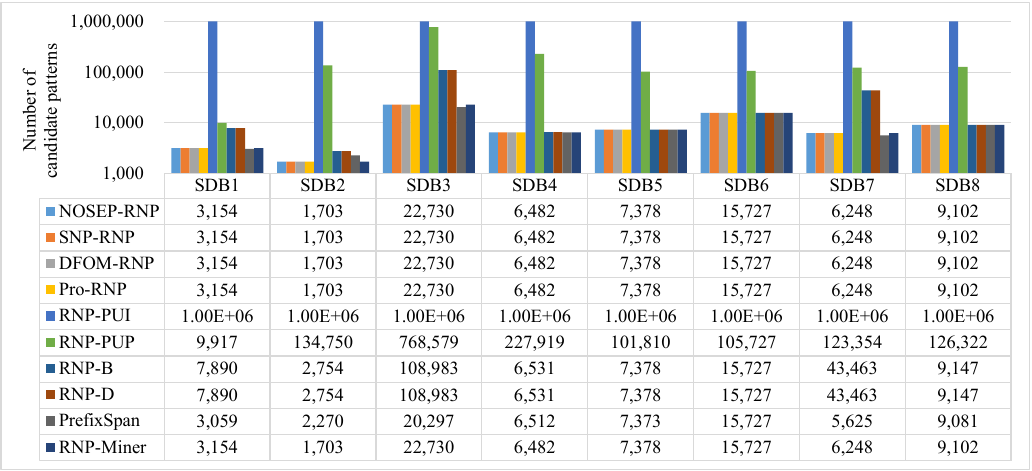"}
    \vspace{-0.9em}
    \caption{Comparison of number of candidate patterns}
    \label{eff-cps}
    \vspace{-1.2em}
\end{figure}
 
\begin{figure}[ht]
    \centering
    \includegraphics[width=0.75\linewidth]{"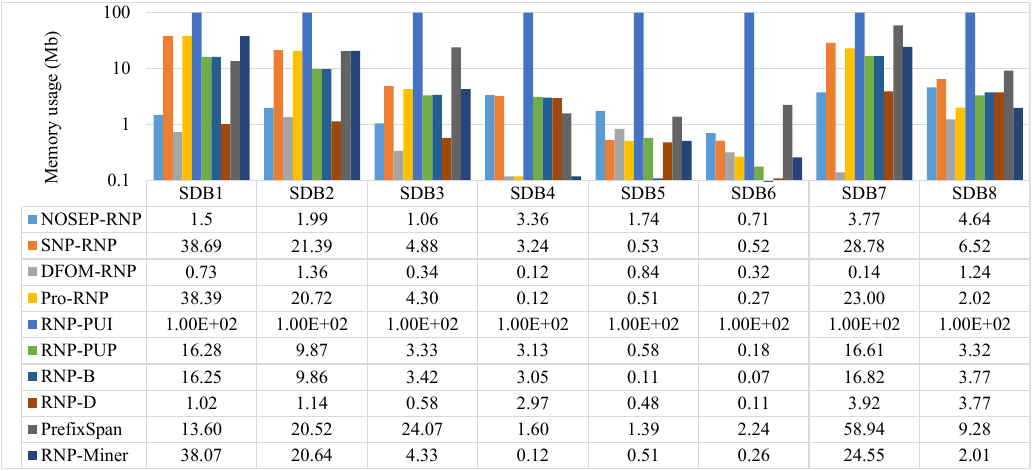"}
    \vspace{-0.9em}
    \caption{Comparison of memory usage}
    \label{eff-memory}
    \vspace{-0.8em}
\end{figure}

The results give rise to the following observations.

% \begin{enumerate} [1.]
1. RNP-Miner runs faster than NOSEP-RNP and SNP-RNP on all databases, which indicates the advantage of using position dictionary. For example, from Fig. \ref{eff-time}, on SDB4, RNP-Miner takes 1.6 s, while NOSEP-RNP and SNP-RNP take 31.17 s and 19.78 s, respectively. Thus, RNP-Miner is about 20 and 12 times faster than NOSEP-RNP and SNP-RNP, respectively. The reason is as follows. The three algorithms adopt different methods to calculate the support. NOSEP-RNP has to create a whole Nettree which contains many useless nodes and parent-child relationships. Therefore, its time complexity is higher than those of SNP-RNP and RNP-Miner. Although SNP-RNP does not create a whole Nettree, it adopts an incomplete Nettree, which can improve the efficiency of the algorithm. However, it needs to scan the original database many times, which is inefficient. Hence, RNP-Miner outperforms NOSEP-RNP and SNP-RNP.

2. RNP-Miner has a better performance than both DFOM-RNP and Pro-RNP, which indicates that the CSC algorithm can efficiently improve the mining performance. Fig. \ref{eff-cps} shows that DFOM-RNP, Pro-RNP, and RNP-Miner mine the same number of candidate patterns, since all three algorithms adopt the same candidate pattern generation method. However, it can be seen that RNP-Miner runs faster than DFOM-RNP and Pro-RNP. For example, on the SDB1 database, DFOM-RNP and Pro-RNP take 34.74 s and 3.70 s, respectively, while RNP-Miner only takes 2.40 s, to find 3,154 candidate patterns. The reason is as follows. RNP-Miner uses the CSC algorithm, which calculates the supports of super-patterns based on the occurrence positions of sub-patterns and does not need to scan the database repeatedly. Therefore, CSC is more efficient than DFOM. Moreover, CSC stores the occurrence positions of patterns with a size of 1, which reduces the number of redundant calculations. Therefore, CSC is more efficient than MatchingPro. Hence, RNP-Miner runs faster than DFOM-RNP and Pro-RNP. However, Fig. \ref{eff-memory} shows that RNP-Miner and Pro-RNP consume more memory than DFOM-RNP. For example, on the SDB1 database, RNP-Miner and Pro-RNP consume 38.07 Mb and 38.39 Mb, respectively, while DFOM-RNP consumes 0.73 Mb. The reason is that both RNP-Miner and Pro-RNP need to store the occurrence positions of sub-patterns, which will lead to an increase in memory usage.

3. RNP-Miner outperforms RNP-PUI and RNP-PUP, which indicates that the pruning strategies can efficiently reduce the number of candidate patterns. Fig. \ref{eff-cps} shows that RNP-Miner generates fewer candidate patterns than RNP-PUI and RNP-PUP, and Fig. \ref{eff-time} shows that RNP-Miner runs faster than RNP-PUI and RNP-PUP. For example, on SDB2, RNP-PUI and RNP-PUP take 1.00E+03 s and 100.15 s for 1.00E+06 and 134,750 candidate patterns, respectively, while RNP-Miner takes 1.56 s for 1,703 candidate patterns. The reason is that the pruning strategies can efficiently prune unpromising items and reduce the number of redundant expansions. We know that the fewer the candidate patterns, the faster the algorithm runs. Hence, RNP-Miner runs faster than RNP-PUI and RNP-PUP.

4. RNP-Miner runs faster than RNP-B and RNP-D, which validates the effectiveness of the itemset pattern join strategy. Fig. \ref{eff-cps} shows that RNP-Miner generates fewer candidate patterns than RNP-B and RNP-D, and Fig. \ref{eff-time} shows that RNP-Miner runs faster than RNP-B and RNP-D. For example, on SDB2, RNP-B takes 4.35 s for 2,754 candidate patterns, and RNP-D takes 4.51 s for 2,754 candidate patterns, while RNP-Miner takes 1.56 s for 1,703 candidate patterns. The reason is as follows. RNP-B and RNP-D use the breadth-first and depth-first strategies based on the enumeration tree, respectively, while RNP-Miner uses the itemset pattern join strategy to generate candidate patterns. As described in the analysis in subsection \ref{subsection:Candidate pattern generation}, the itemset pattern join strategy outperforms the enumeration strategy. Thus, RNP-Miner calculates fewer candidate patterns than RNP-B and RNP-D. Hence, RNP-Miner achieves a better performance.

5. RNP-Miner has a better performance than PrefixSpan. When the numbers of mined patterns by PrefixSpan and RNP-Miner are similar, RNP-Miner runs faster than PrefixSpan on all databases, and consumes less memory than PrefixSpan on SDB2-SDB8. For example, on SDB3, PrefixSpan takes 16.19 s and 24.07 Mb, while RNP-Miner takes 1.90 s and 4.33 Mb. The reason is that PrefixSpan needs to do many database projections, which leads to greater time and memory consumption. Hence, experimental results show that RNP-Miner has the best performance.
% \end{enumerate}

\subsection{Statistical analysis}
\label{subsection:5.3}
To further analyze the significant differences in the efficiency of each algorithm, we use the results in subsection \ref{subsection:5.2} to calculate the Wilcoxon test results between RNP-Miner and NOSEP-RNP, SNP-RNP, DFOM-RNP, Pro-RNP, RNP-PUI, RNP-PUP, RNP-B, RNP-D, and PrefixSpan; the results are shown in Table \ref{t-test}. The \textit{z} statistic of each pair at the 5\% critical level is 2.521, and the corresponding \textit{p} value is 0.012*. This \textit{p} value is less than 0.05. Hence, there are significant differences between RNP-Miner and the other competitive algorithms.

\begin{table}[ht]
	\centering
	% \vspace{-1.2em}
	%\footnotesize{
	\scriptsize
	\caption{Wilcoxon analysis results of paired samples}
	\resizebox{10.5cm}{!}{
		\begin{tabular}{lccccc}
			\hline
			\multirow{2}{*}{Paired samples}	& \multicolumn{2}{c}{$ M(P_{25}, P_{75}) $}	& \multirow{2}{*}{$ |M_{1}-M_{2}|$} & \multirow{2}{*}{\textit{z}} & \multirow{2}{*}{\textit{p}}	\\
			& $ M_{1} (P_{25}, P_{75}) $ & $ M_{2} (P_{25}, P_{75}) $ \\
			\hline
			{NOSEP-RNP \& RNP-Miner} & 71.76(41.7, 97.9) & 1.75(1.5, 2.3) & 69.41 & 2.521 & 0.012* \\
			{SNP-RNP \& RNP-Miner} & 21.22(16.3, 29.8) & 1.75(1.5, 2.3) & 19.47 & 2.521 & 0.012* \\
			{DFOM-RNP \& RNP-Miner} & 6.42(3.0, 11.4) & 1.75(1.5, 2.3) & 4.67 & 2.521 & 0.012* \\
			{Pro-RNP \& RNP-Miner} & 2.87(2.3, 3.4) & 1.75(1.5, 2.3) & 1.12 & 2.521 & 0.012* \\
			{RNP-PUI \& RNP-Miner} & 1,000(1,000, 1,000) & 1.75(1.5, 2.3) & 998.25 & 2.521 & 0.012* \\
			{RNP-PUP \& RNP-Miner} & 77.95(55.2, 159.0) & 1.75(1.5, 2.3) & 76.2 & 2.521 & 0.012* \\
			{RNP-B \& RNP-Miner} & 6.06(2.7, 9.5) & 1.75(1.5, 2.3) & 4.31 & 2.521 & 0.012* \\
			{RNP-D \& RNP-Miner} & 6.35(2.8, 9.8) & 1.75(1.5, 2.3) & 4.6 & 2.521 & 0.012* \\
			{PrefixSpan \& RNP-Miner} & 21.53(18.0,68.8) & 1.75(1.5, 2.3) & 19.78 & 2.521 & 0.012* \\
			\hline
	\end{tabular}}
	\label{t-test}
	% \vspace{-0.3cm}
\end{table}

Moreover, to compare the performances of all algorithms, we utilize Demsar's method \cite{66test2006, 67test2008}. For each database, we rank the competitive algorithms. The fastest one is ranked first, the second-fastest is ranked second, and so on and so forth. The mean ranks of NOSEP-RNP, SNP-RNP, DFOM-RNP, Pro-RNP, RNP-PUI, RNP-PUP, RNP-B, RNP-D, PrefixSpan, and RNP-Miner are 8, 5.5, 4.875, 2.5, 10, 8,125, 3.375, 4.625, 7, and 1, respectively. Now, we use the Friedman test to compare these mean ranks to decide whether or not to reject the null hypothesis. According to $ \tau_{\chi^2}=\dfrac{12n}{k(k+1)}(\sum_{i=1}^k r_i^2-\dfrac{k(k+1)^2}{4}) $ and $ \tau_{F}=\dfrac{(n-1) \tau_{\chi^2}}{n(k-1)-\tau_{\chi^2}} $, the Friedman test result is 28.316, which is greater than $ F_{0.05[9,63]} $ = 2.032. Hence, the null hypothesis of the Friedman test is rejected, that is, there exists a significant difference among the rival algorithms. To further analyze this difference, we use the Nemenyi test, and the results are shown in Fig. \ref{cd}. The dot indicates the mean rank of each scheme, and the black bar indicates the critical difference. Fig. \ref{cd} shows that RNP-Miner is ranked the highest and is significantly better than the other competitive algorithms.

	\begin{figure}[ht]
		\centering
		\includegraphics[width=0.65\linewidth]{"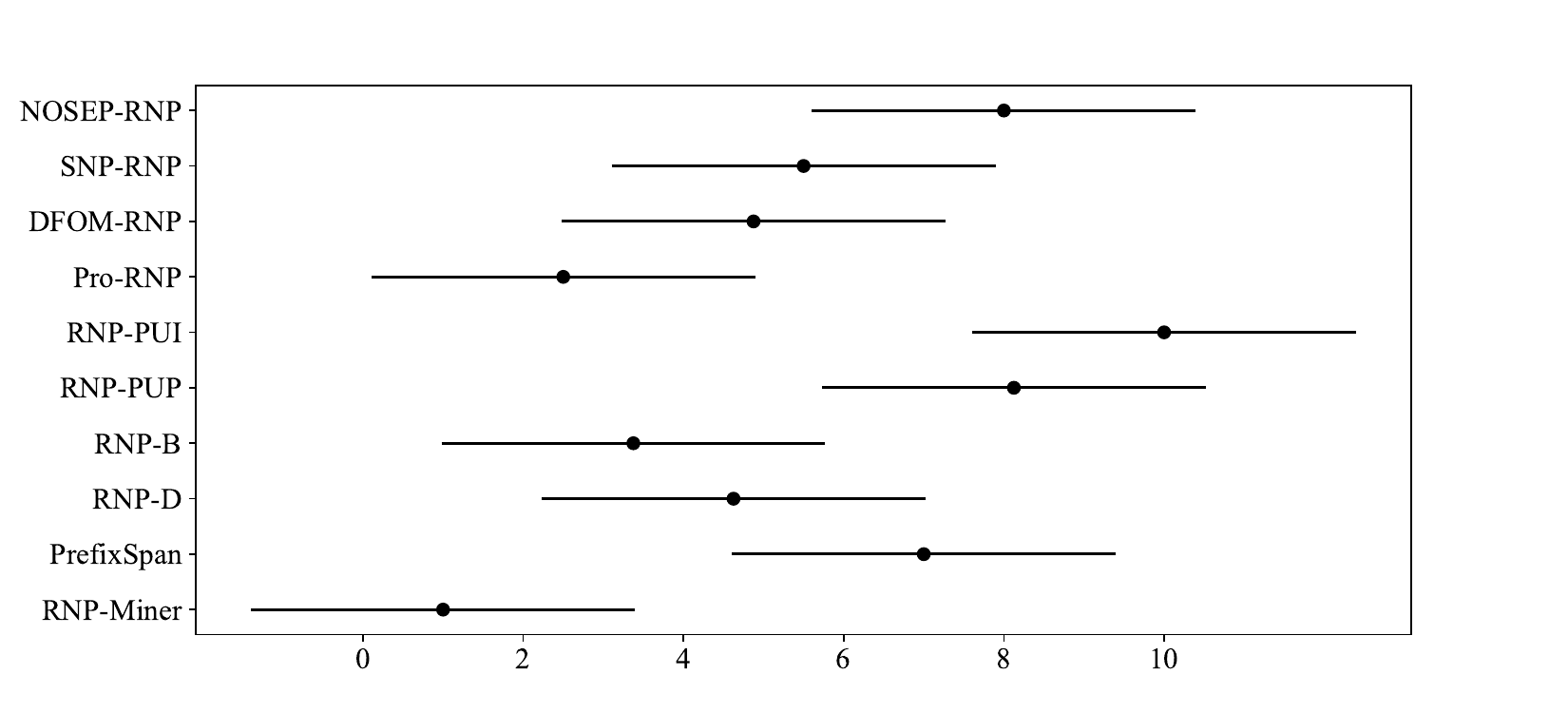"}
		\caption{Nemenyi test of algorithms' mean ranks in terms of the running time}
		\label{cd}
       %     \vspace{-1.2em}
	\end{figure}

\subsection{Influence of different $minsup$ values}
\label{subsection:5.4}
To determine the influence of different \textit{minsup} values on the number of patterns, running time, and memory usage, we select NOSEP-RNP, SNP-RNP, DFOM-RNP, Pro-RNP, RNP-PUI, RNP-PUP, RNP-B, and RNP-D as competitive algorithms, and select SDB2 as the experimental database. We set \textit{minsup} = 1,600, 1,550, 1,500, 1,450, 1,400, 1,350, 1,300, and 1,250, respectively. Since all nine algorithms are complete, the mining results are the same for all of the algorithms, i.e., there are 196, 239, 288, 358, 448, 568, 779, and 942 patterns for  different \textit{minsup} values, respectively. Comparisons of running time, the number of candidate patterns, and memory usage are shown in Figs. \ref{minsup-time}-\ref{minsup-memory}.

\begin{figure}[ht]
		\centering
		\includegraphics[width=0.75\linewidth]{"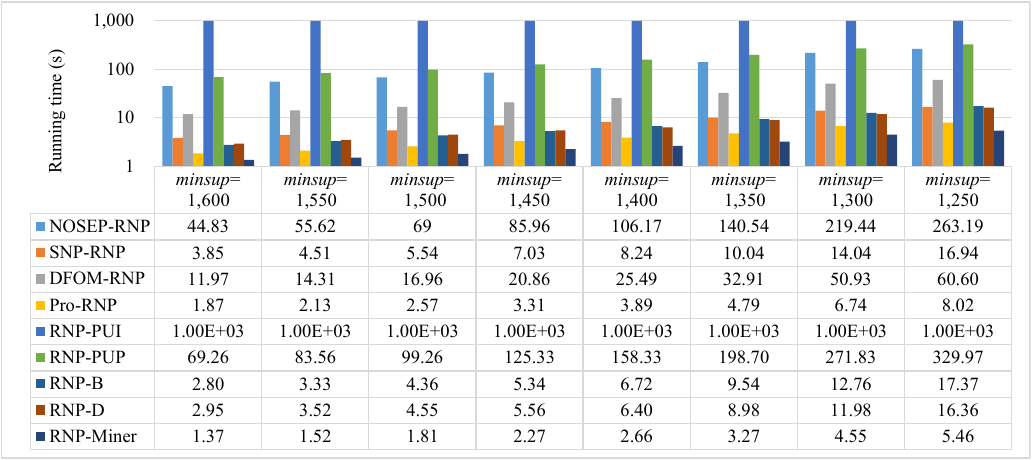"}
		\caption{Comparison of running time for different \textit{minsup} values}
		\label{minsup-time}
  \vspace{-1.2em}
\end{figure}

 \begin{figure}[ht]
		\centering
		\includegraphics[width=0.75\linewidth]{"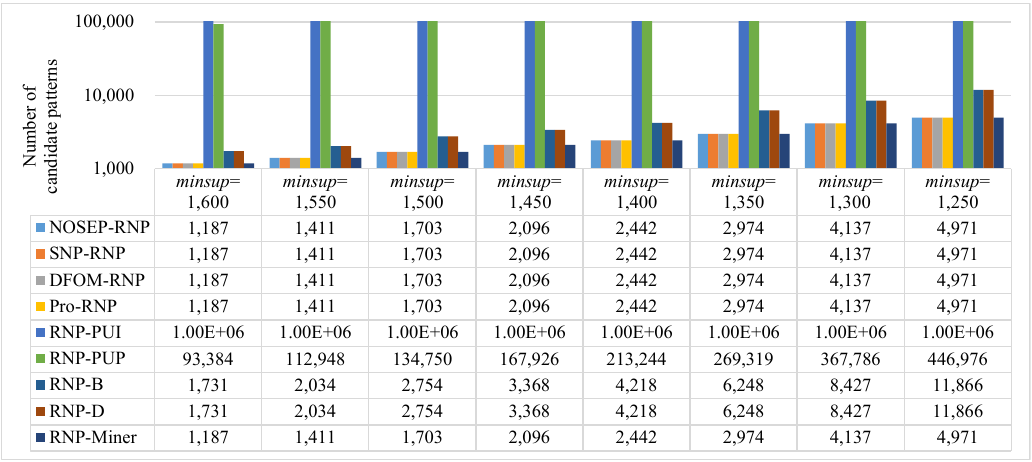"}
		\caption{Comparison of number of candidate patterns for  different \textit{minsup} values}
		\label{minsup-cps}
  \vspace{-1.2em}
\end{figure}

	\begin{figure}[ht]
		\centering
		\includegraphics[width=0.75\linewidth]{"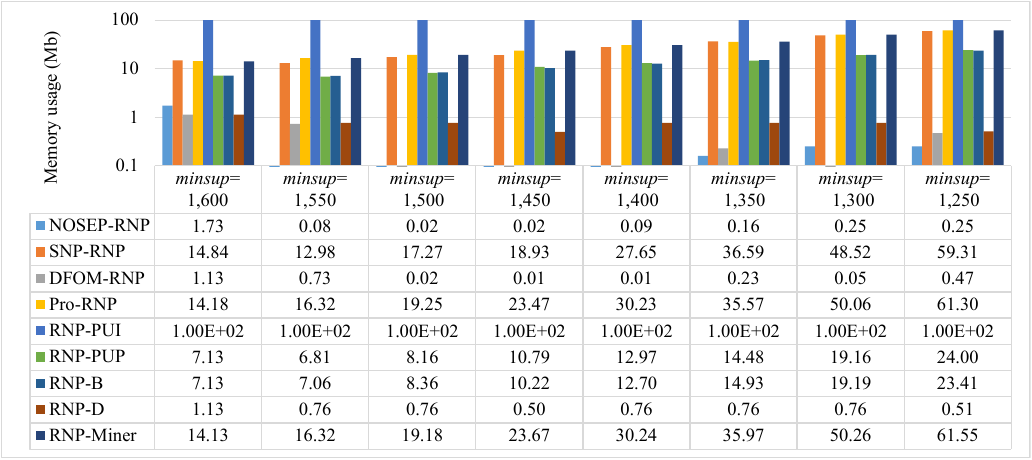"}
		\caption{Comparison of memory usage for different \textit{minsup} values}
		\label{minsup-memory}
    \end{figure}
    
The results give rise to the following observations.

% \begin{enumerate} [1.]
1. Although \textit{minsup} is different, the average running time of each candidate pattern is almost the same.  For example, when \textit{minsup} = 1,600 and \textit{minsup} = 1,350, RNP-Miner checks 1,411 and 2,974 candidate patterns and runs for 1.52 s and 3.27 s, respectively. Thus, when \textit{minsup} = 1,600 and \textit{minsup} = 1,350, the average running time of each candidate pattern is 1.08 ms and 1.10 ms, respectively, which are almost the same.
 
2. As \textit{minsup} decreases, the number of RNPs, the number of candidate patterns, running time, and memory usage increase. For example, when \textit{minsup} = 1,600, RNP-Miner mines 169 RNPs, generates 1,187 candidate patterns, takes 1.37 s, and consumes 14.13 Mb, while when \textit{minsup} = 1,250, RNP-Miner mines 942 RNPs, generates 4,971 candidate patterns, takes 5.46 s, and consumes 61.55 Mb. This phenomenon can also be found in other competitive algorithms. The reason is as follows. As \textit{minsup} decreases, according to Definitions \ref{definition3} and \ref{definition4}, more patterns can be RNPs and candidate patterns. As discussed previously, for different \textit{minsup} values, the average running times of each candidate pattern are almost the same. Hence, the running time increases. Similarly, the memory usage also increases. More importantly, RNP-Miner outperforms other competitive algorithms for all the tested $minsup$ values, which is consistent with the results of subsection \ref{subsection:5.2}.
% \end{enumerate}

\subsection{Scalability}
\label{subsection:5.5}
To evaluate the scalability of RNP-Miner, we select NOSEP-RNP, SNP-RNP, DFOM-RNP, Pro-RNP, RNP-PUI, RNP-PUP, RNP-B, and RNP-D as competitive algorithms. Moreover, we utilize SDB9-SDB16 as the experimental databases; they have sizes of 1k, 5k, 10k, 50k, 100k, 500k, 1,000k, and 5,000k, respectively. Obviously, if \textit{minsup} is constant, the longer the sequence, the more RNPs will be generated. The running time and memory usage are positively correlated with the number of RNPs. To avoid the impact that different numbers of RNPs have on the running time and memory usage, we set \textit{minsup} = 7, 23, 51, 330, 765, 3,950, 7,600, and 38,000 on SDB9-SDB16 and mined 1,065, 795, 748, 748, 765, 737, 728, and 685 patterns, respectively. Comparisons of the running time and memory usage are shown in Figs. \ref{sca-time} and \ref{sca-memory}, respectively.
  
    \begin{figure}[ht]
    \vspace{-0.8em}
		\centering
		\includegraphics[width=0.75\linewidth]{"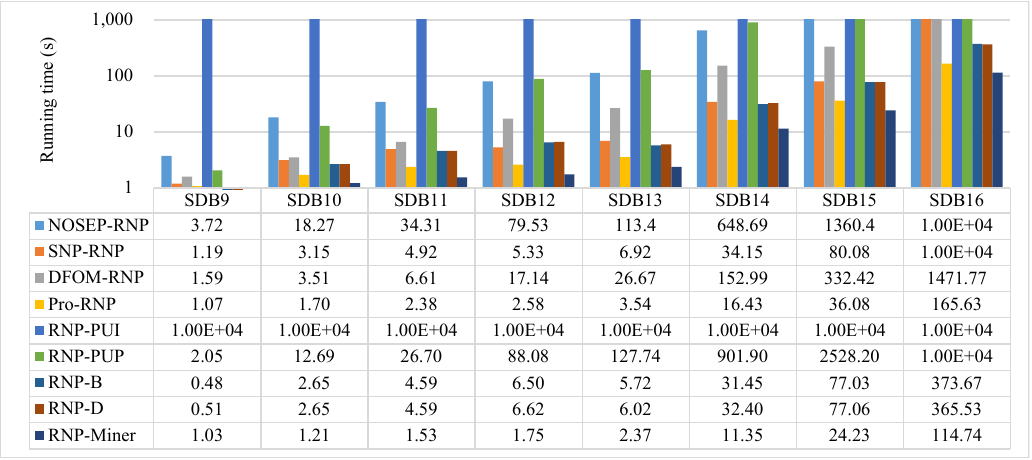"}
		\caption{Comparison of running time for different database sizes}
		\label{sca-time}
  %\vspace{-0.8em}
	\end{figure}
 
	\begin{figure}[ht]
		\centering
		\includegraphics[width=0.75\linewidth]{"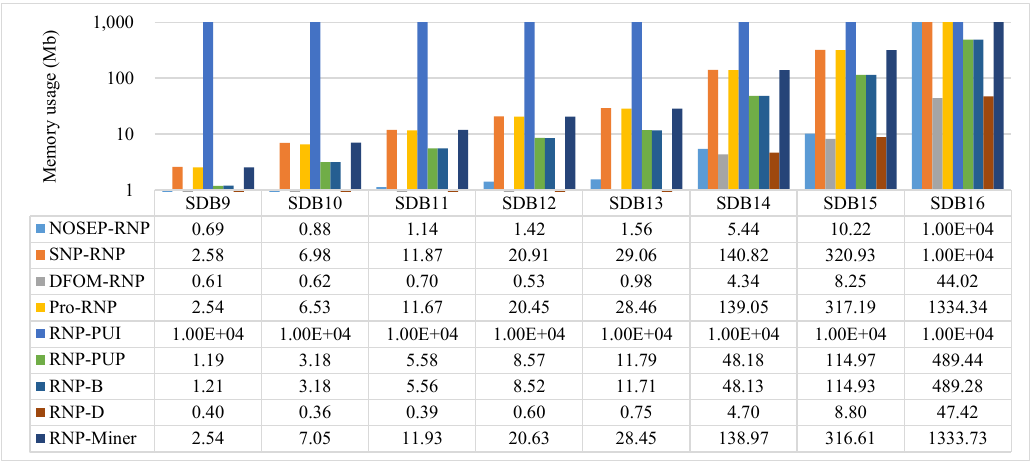"}
		\caption{Comparison of memory usage for different database sizes}
		\label{sca-memory}
  %\vspace{-1.2em}
	\end{figure}

The results give rise to the following observations.

From Figs. \ref{sca-time} and \ref{sca-memory}, the running time and memory usage grow more slowly than the database size. For example, the size of SDB12 is 50 times that of SDB9. RNP-Miner takes 1.75 s and consumes 20.63 Mb on SDB12, which is 1.75/1.03 = 1.699 times the running time and 20.63/2.54 = 8.122 times the memory usage on SDB9. This phenomenon can be found in all other databases. The results indicate that the running time and memory usage are positively correlated with the database size.

Moreover, RNP-Miner has better scalability than other competitive algorithms. When the size of the database increases from 1k to 5,000k, the increase in the running time of RNP-Miner is lower than that of the other algorithms. From Fig. \ref{sca-time}, the increases in the running time for NOSEP-RNP, SNP-RNP, DFOM-RNP, Pro-RNP, RNP-PUP, RNP-B, and RNP-D are 1.00E+04/3.72 = 2688.172, 1.00E+04/1.19 = 8403.361, 1471.77/1.59 = 925.642, 165.63/1.07 = 154.794, 13942.12/2.05 = 6801.034, 373.67/0.48 = 778.479, and 365.53/0.51 = 716.725, respectively, which are higher than that of RNP-Miner (114.74/1.03 = 111.398). However, from Fig. \ref{sca-memory}, the increase in the amount of memory usage for RNP-Miner is 1333.73/2.54 = 525.091, which is greater than the increases in memory usage of the comparison algorithms. The reason is the same as that described in subsection \ref{subsection:5.2}. In summary, RNP-Miner has strong scalability.

\subsection{Case study}
\label{subsection:5.6}
To validate the clustering performance of RNP-Miner, a clustering experiment is conducted in this subsection. We utilized SDB1-8 and SDB17-20 as the raw data to conduct the experiment according to the following steps:

%\begin{enumerate}[Step 1.]
Step 1: We employed CM-SPAM \cite{55CM-SPAM} to mine the classic frequent patterns (FPs) without considering the repetitions of patterns in a sequence and used RNP-Miner to mine the RNPs considering the repetition of patterns and nonoverlapping condition. Since the two algorithms have different definitions of support, they cannot mine the same number of patterns. For fairness, we set two $minsup$ for CM-SPAM to mine slightly fewer and slightly more patterns than RNPs, and the mining results are recorded as FP1s and FP2s, respectively. The setting of $minsup$ and the number of mined patterns in this experiment are shown in Table \ref{case-minsup}.

    \begin{table}[!ht]
	\centering
	\scriptsize
	\caption{The setting of \textit{minsup} and  number of mined patterns}
	\resizebox{12cm}{!}{
		\begin{tabular}{lcccccc}
			\hline
			\multirow{2}{*}{}	& \multicolumn{4}{c}{CM-SPAM}	& \multicolumn{2}{c}{RNP-Miner} \\
                {} & \textit{minsup} & Number of FP1s & \textit{minsup} & Number of FP2s & \textit{minsup} & Number of RNPs \\
			\hline
			{SDB1} & 0.98	& 77	& 0.93	& 287	& 1,200	& 165\\
			{SDB2} & 0.50	& 137	& 0.45	& 293	& 1,500	& 288\\
                {SDB3} & 0.95	& 478	& 0.93	& 1100	& 170	& 736\\
                {SDB4} & 0.20	& 43	& 0.19	& 52	& 200	& 50 \\
                {SDB5} & 0.50	& 2	& 0.40	& 18	& 3,000	& 8\\
                {SDB6} & 0.25	& 4	& 0.20	& 8	& 2,000	& 5\\
                {SDB7} & 0.44	& 332	& 0.38	& 636	& 300	& 461\\
                {SDB8} & 0.43	& 10	& 0.3	& 22	& 400	& 13\\
                {SDB17} & 0.96	& 216	& 0.91	& 312	& 150	& 273\\
                {SDB18} & 0.98	& 63	& 0.90	& 185	& 35	& 165\\
                {SDB19} & 0.95	& 74	& 0.89	& 116	& 6,000	& 101\\
                {SDB20} & 0.84	& 79	& 0.77	& 132	& 9,000	& 107\\
			\hline
	    \end{tabular}}
	\label{case-minsup}
    \end{table}

Step 2: The K-means++ method is adopted to cluster the Raw data, FP1s, FP2s, and RNPs, respectively.
 
Step 3: For databases without actual classification (SDB1-8), we use two metrics to determine the clustering performance: Silhouette Coefficient (\textit{SC}) and Calinski-Harabaz Index (\textit{CHI}), which can be calculated using Equations \ref{eq:SC} and \ref{eq:CHI}, respectively:

\begin{small}
    \begin{equation}
     SC=\dfrac{1}{N}\sum_{i=1}^{N} \dfrac{b(i)-a(i)}{max \{a(i),b(i)\}}
     \label{eq:SC}
    \end{equation}
\end{small}
            
\begin{small}
    \begin{equation}
        CHI=\dfrac{T_r(B_K)}{T_r(W_K)} \times \dfrac{N-K}{K-1}
    \label{eq:CHI}
    \end{equation}
\end{small}
            
        where \textit{N} is the number of samples. The experiments are conducted with \textit{K} = 2,3,4,$\dots$, 10. The results are shown in Tables \ref{t-SC} and \ref{t-CHI}.

\begin{table}[!ht]
    \centering
    \scriptsize
    \caption{Comparison of \textit{SC} on the databases without actual classification}
    \vspace{2pt}
    \resizebox{11cm}{!}{
    \begin{tabular}{lccccccccc}
                    \hline
                    \multicolumn{2}{c}{} & \textit{SDB1} &  \textit{SDB2} & \textit{SDB3} & \textit{SDB4} & \textit{SDB5} & \textit{SDB6} & \textit{SDB7} &\textit{SDB8} \\
                    \hline
    \multirow{4}{*}{\textit{K}=2} & {Raw} & 0.06 & 0.29 & 0.07 & 0.25 & 0.43 & 0.43 & 0.11 & 0.34\\
     & {FP1s} & 0.85 & 0.53 & 0.39 & 0.19 & 0.67 & 0.78 & 0.07 & 0.27\\
     & {FP2s} & 0.84	& 0.57 & 0.37 & 0.15 & 0.64	& 0.64 & 0.11 & 0.39\\
     & RNPs & \textbf{0.92} & \textbf{0.74} & \textbf{0.55} & \textbf{0.95} & \textbf{0.94} & \textbf{0.94} & \textbf{0.13} & \textbf{0.89}\\
    
     \hline
     \multirow{4}{*}{\textit{K}=3} & Raw	& 0.02 & 0.23 & 0.11 & 0.14 & 0.41 & 0.33 & \textbf{0.10} & 0.33\\
     & FP1s & 0.47 & 0.50 & 0.25 & 0.12 & {-} & 0.77 & 0.06 & 0.16\\
     & FP2s & 0.43 & 0.48 & 0.24 & 0.08 & 0.58 & 0.72 & \textbf{0.10} & 0.29\\
     & RNPs & \textbf{0.84} & \textbf{0.62} & \textbf{0.49} & \textbf{0.95} & \textbf{0.85} & \textbf{0.88} & \textbf{0.10} & \textbf{0.87}\\
    
                    \hline                              
    \multirow{4}{*}{\textit{K}=4} & Raw	& 0.01	& 0.16	& 0.10	& 0.12	& 0.37	& 0.25	& \textbf{0.09}	& 0.24\\
     & FP1s & 0.36 & 0.52 & 0.21 & 0.13 & - & 0.77 & 0.05 & 0.17\\
     & FP2s & 0.36 & 0.49 & 0.21 & 0.07 & 0.64 & 0.80 & \textbf{0.09} & 0.28\\
     & RNPs & \textbf{0.82} & \textbf{0.55} & \textbf{0.46} & \textbf{0.29} & \textbf{0.79} & \textbf{0.86} & \textbf{0.09} & \textbf{0.89}\\

                    \hline
                    \multirow{4}{*}{\textit{K}=5} & Raw & 0.02 & 0.17 & 0.09 & 0.13 & 0.33 & 0.23 & \textbf{0.09} & 0.16\\
      & FP1s & 0.37 & 0.53 & 0.19 & 0.08 & - & - & 0.05 & 0.09\\
     & FP2s & 0.36 & 0.50 & 0.17 & 0.06 & \textbf{0.78} & 0.74 & \textbf{0.09} & 0.19\\
      & RNPs & \textbf{0.82} & \textbf{0.56} & \textbf{0.46} & \textbf{0.27} & 0.76 & \textbf{0.75} & \textbf{0.09} & \textbf{0.78}\\
                    \hline
     \multirow{4}{*}{\textit{K}=6} & Raw & 0.03 & 0.17 & 0.13 & 0.10 & 0.33 & 0.18 & 0.10 & 0.13\\
      & FP1s & 0.31 & 0.49 & 0.20 & 0.08 & - & - & 0.05 & 0.06\\
     & FP2s & 0.30 & 0.48 & 0.19 & 0.06 & 0.76 & 0.71 & 0.10 & 0.17\\
     & RNPs & \textbf{0.79} & \textbf{0.56} & \textbf{0.48} & \textbf{0.13} & \textbf{0.85} & \textbf{0.74} & \textbf{0.11} & \textbf{0.73}\\
                    \hline
      \multirow{4}{*}{\textit{K}=7} & Raw & 0.02 & 0.16 & 0.05 & 0.07 & 0.32 & 0.19 & 0.07 & 0.12\\
     & FP1s & 0.32 & 0.49 & 0.20 & 0.08 & - & - & 0.05 & 0.06\\
                                                  & FP2s & 0.30 & 0.49 & 0.19 & 0.05 & 0.72 & 0.67 & \textbf{0.10} & 0.17\\
                                                  & RNPs & \textbf{0.81} & \textbf{0.53} & \textbf{0.46} & \textbf{0.20} & \textbf{0.91} & \textbf{0.74} & \textbf{0.10} & \textbf{0.73}\\
                                                
                    \hline
                    \multirow{4}{*}{\textit{K}=8} & Raw & 0.02 & 0.16 & 0.02 & 0.08 & 0.32 & 0.16 & 0.07 & 0.12\\
                                                  & FP1s & 0.30 & 0.42 & 0.10 & 0.08 & - & - & 0.05 & 0.07\\
                                                  & FP2s & 0.24 & 0.48 & 0.11 & 0.05 & 0.68 & 0.66 & \textbf{0.10} & 0.15\\
                                                  & RNPs & \textbf{0.70} & \textbf{0.52} & \textbf{0.45} & \textbf{0.16} & \textbf{0.96} & \textbf{0.74} & \textbf{0.10} & \textbf{0.68}\\
                    \hline
                    \multirow{4}{*}{\textit{K}=9} & Raw & 0.03 & 0.11 & 0.06 & 0.07 & 0.32 & 0.13 & 0.07 & 0.11\\
                                                  & FP1s & 0.26 & 0.42 & 0.15 & 0.08 & - & - & 0.05 & 0.05\\
                                                  & FP2s & 0.22 & 0.45 & 0.12 & 0.05 & 0.64 & - & \textbf{0.10} & 0.15\\
                                                  & RNPs & \textbf{0.71} & \textbf{0.49} & \textbf{0.46} & \textbf{0.16} & \textbf{0.86} & \textbf{0.72} & \textbf{0.10} & \textbf{0.70}\\
        
                    \hline
                    \multirow{4}{*}{\textit{K}=10} & Raw & 0.02 & 0.12 & 0.05 & 0.07 & 0.32 & 0.12 & 0.07 & 0.11\\
                                                   & FP1s & 0.26 & 0.46 & 0.13 & 0.07 & - & - & 0.05 & 0.06\\
                                                   & FP2s  & 0.22 & 0.38 & 0.11 & 0.04 & 0.61 & - & 0.09 & 0.15\\
                                                   & RNPs & \textbf{0.71} & \textbf{0.50} & \textbf{0.37} & \textbf{0.12} & \textbf{0.89} & \textbf{0.72} & \textbf{0.10} & \textbf{0.72}\\
        
                    \hline
        	\end{tabular}}
            \label{t-SC}
        	% \vspace{-0.3cm}
        \flushleft
	\scriptsize
	Note: - indicates that the result cannot be obtained.
        \end{table}

        \begin{table}[ht]
            \centering
            % \vspace{-1.2em}
            %\footnotesize{
            \scriptsize
            \caption{Comparison of \textit{CHI} on the databases without actual classification}
            \vspace{2pt}
            \resizebox{11cm}{!}{
            \begin{tabular}{lccccccccc}
                \hline
                \multicolumn{2}{c}{} & \textit{SDB1} &  \textit{SDB2} & \textit{SDB3} & \textit{SDB4} & \textit{SDB5} & \textit{SDB6} & \textit{SDB7} &\textit{SDB8} \\
                \hline
                \multirow{4}{*}{\textit{K}=2} & Raw & 14 & 335 & 5 & 101 & 306 & 306 & 92 & 1607\\
                                              & FP1s & 432 & 1025 & 24 & 127 & 2685 & \textbf{1895} & 56 & 691\\
                                              & FP2s & 391 & \textbf{1297} & 36 & 125 & 2692 & 1789 & 92 & 956\\
                                              & RNPs & \textbf{455} & 1164 & \textbf{44} & \textbf{428} & \textbf{3153} & 1652 &\textbf{104} & \textbf{3100}\\
        
                \hline
                \multirow{4}{*}{\textit{K}=3} & Raw & 12 & 222 & 5 & 66 & 805 & 254 & 71 & 1261\\
                                              & FP1s & 531 & 785 & 17 & 91 & - & 2044 & 42 & 440\\
                                              & FP2s & 293 & 1902 & 24 & 86 & 2996 & 1975 & 71 & 617\\
                                              & RNPs & \textbf{598} & \textbf{1099} & \textbf{35} & \textbf{820} & \textbf{3813} & \textbf{2203} &\textbf{84} & \textbf{2530}\\
        
                \hline                               
                \multirow{4}{*}{\textit{K}=4} & Raw & 10 & 174 & 5 & 53 & 821 & 230 & 59 & 949\\
                                              & FP1s & 431 & 719 & 14 & 72 & - & 3713 & 36 & 331\\
                                              & FP2s & 234 & 853 & 19 & 69 & 3431 & 3711 & \textbf{59} & 463\\
                                              & RNPs & \textbf{589} & \textbf{1027} & \textbf{31} & \textbf{570} & \textbf{6267} & \textbf{3734} & \textbf{59} & \textbf{1953}\\
                                            
                \hline
                \multirow{4}{*}{\textit{K}=5} & Raw & 9 & 150 & 5 & 45 & 777 & 198 & 53 & 746\\
                                              & FP1s & 378 & 625 & 12 & 64 & - & - & 31 & 270\\
                                              & FP2s & 199 & 759 & 16 & 58 & 4532 & 4388 & 53 & 374\\
                                              & RNPs & \textbf{604} & \textbf{995} & \textbf{25} & \textbf{424} & \textbf{7134} & \textbf{4426} & \textbf{62} & \textbf{1550}\\
                                            
                \hline
                \multirow{4}{*}{\textit{K}=6} & Raw & 10 & 132 & 5 & 37 & 742 & 178 & 49 & 620\\
                                              & FP1s & 349 & 557 & 12 & 56 & - & - & 28 & 231\\
                                              & FP2s & 176 & 679 & 15 & 52 & 6317 & 4526 & 49 & 321\\
                                              & RNPs & \textbf{624} & \textbf{1008} & \textbf{22} & \textbf{348} & \textbf{7955} & \textbf{4806} & \textbf{54} & \textbf{1297}\\
                                            
                \hline
                \multirow{4}{*}{\textit{K}=7} & Raw & 9 & 115 & 4 & 35 & 707 & 163 & 46 & 533\\
                                              & FP1s & 333 & 491 & 11 & 50 & - & - & 26 & 202\\
                                              & FP2s & 161 & 597 & 13 & 46 & 10202 & 5819 & 46 & 283\\
                                              & RNPs & \textbf{645} & \textbf{1013} & \textbf{20} & \textbf{295} & \textbf{10317} & \textbf{6056} & \textbf{48} & \textbf{1132}\\

                \hline
                \multirow{4}{*}{\textit{K}=8} & Raw & 8 & 103 & 4 & 32 & 664 & 149 & 43 & 470\\
                                              & FP1s & 326 & 449 & 11 & 56 & - & - & 24 & 180\\
                                              & FP2s & 154 & 537 & 13 & 43 & 10588 & 5962 & 43 & 254\\
                                              & RNPs & \textbf{615} & \textbf{988} & \textbf{19} & \textbf{257} & \textbf{11241} & \textbf{6592} & \textbf{44} & \textbf{1002}\\
        
                \hline
                \multirow{4}{*}{\textit{K}=9} & Raw & 8 & 95 & 4 & 29 & 643 & 137 & 40 & 424\\
                                              & FP1s & 318 & 411 & 11 & 43 & - & - & 23 & 164\\
                                              & FP2s & 149 & 491 & 13 & 39 & 10845 & - & 40 & 229\\
                                              & RNPs & \textbf{589} & \textbf{975} & \textbf{18} & \textbf{226} & \textbf{11797} & \textbf{7999} & \textbf{41} & \textbf{903}\\
        
                \hline
                \multirow{4}{*}{\textit{K}=10} & Raw & 8 & 88 & 4 & 27 & 620 & 127 & \textbf{38} & 386\\
                                               & FP1s & 326 & 385 & 11 & 40 & - & - & 21 & 151\\
                                               & FP2s & 148 & 460 & 13 & 35 & 11291 & - & \textbf{38} & 211\\
                                               & RNPs & \textbf{571} & \textbf{952} & \textbf{17} & \textbf{204} & \textbf{12893} & \textbf{9345} & \textbf{38} & \textbf{827}\\

                \hline
            \end{tabular}}
            \label{t-CHI}
            % \vspace{-0.3cm}
        \flushleft
	\scriptsize
	Note: - indicates that the result cannot be obtained.
% \vspace{-2em}
        \end{table}

Step 4: For databases with actual classification (SDB17-20), we use two metrics to evaluate the clustering performance: normalized mutual information (\textit{NMI}) \cite{68nmi} and homogeneity (\textit{h}) \cite{69h} which can be calculated using Equations \ref{eq:NMI} and \ref{eq:h}, respectively. The results are shown in Table \ref{t-NMI&h}.

        \begin{small}
        \begin{equation}
        NMI(X,Y)=\dfrac{\sum_{i=1}^{|X|} \sum_{j=1}^{|Y|} P(i,j) {\rm log} (\dfrac{P(i,j)}{P(i)P(j)})}{\sqrt{\sum_{i=1}^{|X|} P(i) {\rm log} P(i) \times \sum_{j=1}^{|Y|} P(j) {\rm log} P(j)}}
        \label{eq:NMI}
        \end{equation}
        \end{small}
        \begin{small}
        \begin{equation}
        h(X,Y)=1-\dfrac{-\sum_{i=1}^{|X|} \sum_{j=1}^{|Y|} P(i,j) {\rm log}P(i|j)}{- \sum_{i=1}^{|X|} P(i) {\rm log}P(i)}
        \label{eq:h}
        \end{equation}
        \end{small}

    \begin{table}[ht]
    	\centering
    	% \vspace{-1.2em}
    	%\footnotesize{
    	\scriptsize
    	\caption{Comparison of \textit{NMI} and $h$ on the databases with actual classification}
        \vspace{2pt}
    	\resizebox{9cm}{!}{
    		\begin{tabular}{lcccccccc}
    			\hline
    			\multirow{2}{*}{}	& \multicolumn{2}{c}{SDB17}	& \multicolumn{2}{c}{SDB18}	& \multicolumn{2}{c}{SDB19}	& \multicolumn{2}{c}{SDB20}	\\
                {} & \textit{NMI} & \textit{h} & \textit{NMI} & \textit{h} & \textit{NMI} & \textit{h} & \textit{NMI} & \textit{h} \\
    		\hline
    			{Raw} &	0.05	&	0.05	&	0.01	&	0.01	&	0.28	&	0.23	&	0.03	&	0.06 \\
    			{FP1s} &	0.20	&	0.15	&	0.08	&	0.06	&	0.35	&	0.32	&	0.15	&	0.14 \\
                    {FP2s} & 0.29	&	0.23	&	0.48	&	0.44	&	0.76	&	0.73	&	0.26	&	0.26 \\
    			{RNPs} &	\textbf{0.75}	&	\textbf{0.74}	&	\textbf{0.73}	&	\textbf{0.70}	&	\textbf{0.86}	&	\textbf{0.86}	&	\textbf{0.92}	&	\textbf{0.89} \\			
    			\hline
    	\end{tabular}}
    	\label{t-NMI&h}
    \end{table}
    
%\end{enumerate}

The results give rise to the following observations.

The experimental results on databases with and without actual classification indicate that the clustering performance using Raw is the worst, the clustering performances using FP1s and FP2s are better than Raw, and the clustering performance using RNPs is the best.  For example, on the databases without actual classification, \textit{SC} and \textit{CHI} are two metrics for evaluating clustering performance. The range of \textit{SC} is [-1,1]. The closer the value is to 1, the better the clustering performance. \textit{CHI} has no upper bound. The higher the value, the better the clustering performance. As shown in Tables \ref{t-SC} and \ref{t-CHI}, regardless of the value of \textit{K}, the \textit{SC} and \textit{CHI} values of RNPs are higher than those of Raw and FPs in most cases. For example, on SDB1, when \textit{K}=2, the \textit{SC} and \textit{CHI} values of the Raw data clustering results are 0.06 and 14, respectively,  \textit{SC} and \textit{CHI} of FP1s are 0.85 and 432, respectively,  \textit{SC} and \textit{CHI} of FP2s are 0.84 and 391, respectively, while \textit{SC} and \textit{CHI} of RNPs are 0.92 and 455, respectively.

Moreover, on the databases with actual classification, \textit{NMI} and \textit{h} are two metrics for evaluating clustering performance, these reflect the similarity between the clustering results and the actual values. The greater the \textit{NMI} and \textit{h} values, the more similarity there is between the clustering results and the actual values. As shown in Table \ref{t-NMI&h}, no matter which databases, the \textit{NMI} and \textit{h} values of RNPs are higher than those of the Raw data, FP1s, and FP2s. For example, on SDB17, the \textit{NMI} and \textit{h} values of the Raw data clustering results are 0.05 and 0.05, respectively, the \textit{NMI} and \textit{h} values of the FP1s are 0.20 and 0.15, respectively, the \textit{NMI} and \textit{h} values of the FP2s are 0.29 and 0.23, respectively, while the \textit{NMI} and \textit{h} values of RNPs are 0.75 and 0.74, respectively. This phenomenon can also be found in all the other databases. The reason is as follows. 
% \begin{enumerate} [1.]

1. It is difficult to obtain good clustering results using the raw data, since the Raw data generally contains much redundant information, which will affect the clustering performance. Hence, the clustering performance using Raw is the worst.

2. Classical SPM determines whether a pattern occurs in a sequence or not. A sequence represents a user's sequential purchase behavior. Thus, the frequent patterns are the common purchasing behaviors of many users, which can reflect their common interests. Therefore, using the classical SPM method to extract features provides a better clustering performance than clustering the original data. Hence, the clustering performances using FP1s and FP2s are better than Raw. 

3. Nevertheless, some users' interests are completely different. For example, the interests of users who only buy once and the interests of those who buy many times are obviously different. However, using classical SPM methods, users cannot know the number of occurrences of a pattern in a sequence, and they only know whether a pattern occurs in a sequence or not. Hence, classical SPM cannot reflect the number of occurrences, i.e., the degree of interest.

4. Our mining method is a repetitive SPM method, and it calculates the number of occurrences, which represents the degree of interest. Obviously, we can obtain a better clustering performance by using the user's degree of interest. Therefore, using our method to extract features provides a better clustering performance than using the classical SPM method. Hence, our mining method outperforms classical SPM in feature extraction. %, i.e., the clustering performance using RNPs has the best performance.
% \end{enumerate}

\subsection{Engineering Applications}
\label{subsection:Engineering Applications}
The proposed mining method can be used in many applications, as shown in Fig. \ref{appli}, such as purchase behavior analysis, web intrusion detection, and biological sequence analysis. For example, B2C e-commerce customer behavior patterns based on user purchase records can be analyzed to develop reasonable marketing strategies. Compared with traditional SPM methods, RNP-Miner can detect customer repetitive purchase behaviors, and merchants can develop more targeted strategies based on whether customers repurchase and the number of times they repurchase. In biological sequence analysis, RNP-Miner can also discover information overlooked by traditional biological sequence mining through the repetition of pattern in each sequence. The application of repetitive pattern mining methods in genomic data helps to discover how many different repetitive sequences are presented in the genome, and how their number of replicates varies in different organisms, thus identifying the relationship between repetitive sequences and diseases. For example, the number of copies of the gene named CCL3L1 in different races is different. The more the number of this gene, the stronger resistant it is to HIV.
\begin{figure}[ht]
		\centering
		\includegraphics[width=0.6\linewidth]{"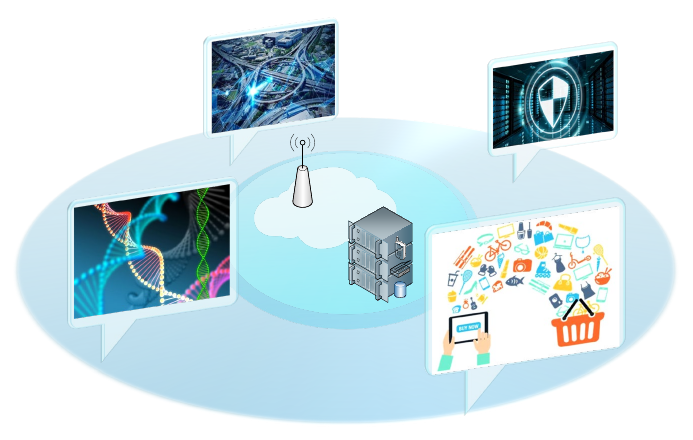"}
		\caption{The applications of the proposed mining method}
		\label{appli}
  \vspace{-1.2em}
	\end{figure}

\section{CONCLUSION}
\label{section:CONCLUSION}
Classical SPM only considers whether or not a given pattern occurs within a sequence and ignores the repetition of the pattern in this sequence. To overcome this drawback, this paper investigates repetitive SPM, called RNP mining, which considers the repetition of patterns in each sequence. To effectively discover all frequent RNPs, we propose the RNP-Miner algorithm, which consists of three parts: 1) data preprocessing, 2) candidate pattern generation, and 3) support calculation. In data preprocessing, the position dictionary structure can improve the algorithm's efficiency by reducing the number of scans of the original database. In candidate pattern generation, the enumeration strategy based on the I-extension and S-extension can generate candidate patterns. However, it generates many redundant candidate patterns. To improve efficiency, we proposed an itemset pattern join strategy based on the Apriori property. In support calculation, the repeated scanning of the database will take a large amount of time. Thus, we proposed the CSC algorithm, which calculates the supports of super-patterns based on the occurrence positions of sub-patterns and new extension itemsets. We evaluate the performance of RNP-Miner using real-life sequential databases, such as a transaction database. The extensive experimental results demonstrate that RNP-Miner not only has a better running performance than other competitive algorithms but also can discover the patterns that occur in a sequence repetitively. More importantly, RNP mining implements feature extraction, which improves the clustering performance.	
	
\section*{Acknowledgement}
This work was partly supported by the National Natural Science Foundation of China (62372154, 61976240, 62120106008), and in part by the Graduate Student Innovation Program of Hebei Province under Grant CXZZBS2022030.

\end{document}